\documentclass{sig-alternate-10pt}

\usepackage{balance}
\usepackage{float}

\providecommand{\tabularnewline}{\\}
\floatstyle{ruled}
\newfloat{algorithm}{tbp}{loa}
\providecommand{\algorithmname}{Algorithm}
\floatname{algorithm}{\protect\algorithmname}

\begin{document}
%

\pretolerance=2000

\title{Streaming Video over HTTP with Consistent Quality}
%
%
%
%
%

\numberofauthors{6} 
%
\author{
%
%
\alignauthor
Zhi Li\\
       \affaddr{Cisco}\\
       \affaddr{170 West Tasman Dr.}\\
       \affaddr{San Jose, CA 95134, USA}\\
       \email{zhil2@cisco.com}
\alignauthor
Ali C. Begen\\
       \affaddr{Cisco Canada}\\
       \affaddr{181 Bay St., Suite 3400}\\
       \affaddr{Toronto, ON M5J 2T3, Canada}\\
       \email{abegen@cisco.com}
\alignauthor 
Joshua Gahm\\
       \affaddr{Cisco}\\
       \affaddr{1414 Massachusetts Ave.}\\
       \affaddr{Boxborough, MA 01719, USA}\\
       \email{jgahm@cisco.com}
\and  
\alignauthor Yufeng Shan\\
       \affaddr{Cisco}\\
       \affaddr{1414 Massachusetts Ave.}\\
       \affaddr{Boxborough, MA 01719, USA}\\
       \email{yshan@cisco.com}
\alignauthor Bruce Osler\\
       \affaddr{Cisco}\\
       \affaddr{1414 Massachusetts Ave.}\\
       \affaddr{Boxborough, MA 01719, USA}\\
       \email{brosler@cisco.com}
\alignauthor David Oran\\
       \affaddr{Cisco}\\
       \affaddr{55 Cambridge Pkwy, Suite 101}\\
       \affaddr{Cambridge, MA 02142, USA}\\
       \email{oran@cisco.com}
}


\maketitle
\begin{abstract}
In conventional HTTP-based adaptive streaming (HAS), a video source is encoded at multiple levels of constant bitrate representations, and a client makes its representation selections according to the measured network bandwidth. While greatly simplifying adaptation to the varying network conditions, this strategy is not the best for optimizing the video quality experienced by end users. Quality fluctuation can be reduced if the natural variability of video content is taken into consideration. In this work, we study the design of a client rate adaptation algorithm to yield consistent video quality. We assume that clients have visibility into incoming video within a finite horizon. We also take advantage of the client-side video buffer, by using it as a breathing room for not only network bandwidth variability, but also video bitrate variability. The challenge, however, lies in how to balance these two variabilities to yield consistent video quality without risking a buffer underrun. We propose an optimization solution that uses an online algorithm to adapt the video bitrate step-by-step, while applying dynamic programming at each step. We incorporate our solution into PANDA -- a practical rate adaptation algorithm designed for HAS deployment at scale.
\end{abstract}




\section{Introduction}

Over the past few years, we have witnessed that streaming video over
the Internet is converging towards a new paradigm named HTTP-based
adaptive streaming (HAS), also dubbed as dynamic adaptive streaming
over HTTP (DASH).

In an HAS system, a video source is chopped into short chunks of a
few seconds each (which we will also refer to as segment in this paper).
Every segment is independently encoded (or transcoded from a single
master high-quality source) at several different bitrates, and the
output representations are stored at a server from which clients fetch
the segments. Common practice is for the encoder/transcoder to employ
constant bitrate (CBR), resulting in a set of tiers, or ``levels''
of video output. A client application fetches the segments from the
server sequentially using plain HTTP GETs, estimates the available
bandwidth using measurements of the downloading performance, and adapts
the level selection of the next segment to fetch at the completion
of the current segment. Typically, tens of seconds of downloaded content
are buffered at the client to accommodate bandwidth variability. A
viable client rate adaptation algorithm must fetch the video segments
to make best use of the available bandwidth, while without risking
to drain the client buffer and causing video playout stalls.

\begin{figure}
\begin{centering}
\hspace{-0in} 
\par\end{centering}

\begin{centering}
\begin{minipage}[t]{1\columnwidth}%
\begin{center}
\vspace{-0.22in}

\par\end{center}

\begin{center}
\includegraphics[scale=0.16]{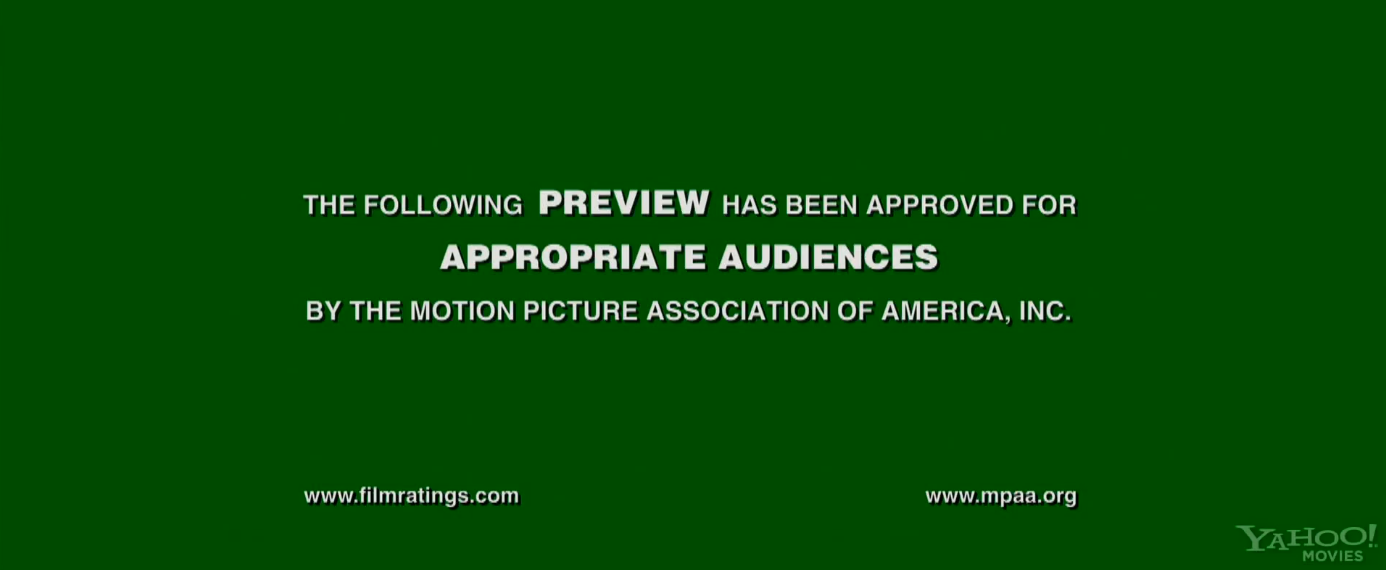} 
\par\end{center}%
\end{minipage}
\par\end{centering}

\begin{centering}
\vspace{0in}
 \hspace{-0in} 
\par\end{centering}

\begin{centering}
\begin{minipage}[t]{1\columnwidth}%
\begin{center}
\vspace{-0.22in}

\par\end{center}

\begin{center}
\includegraphics[scale=0.16]{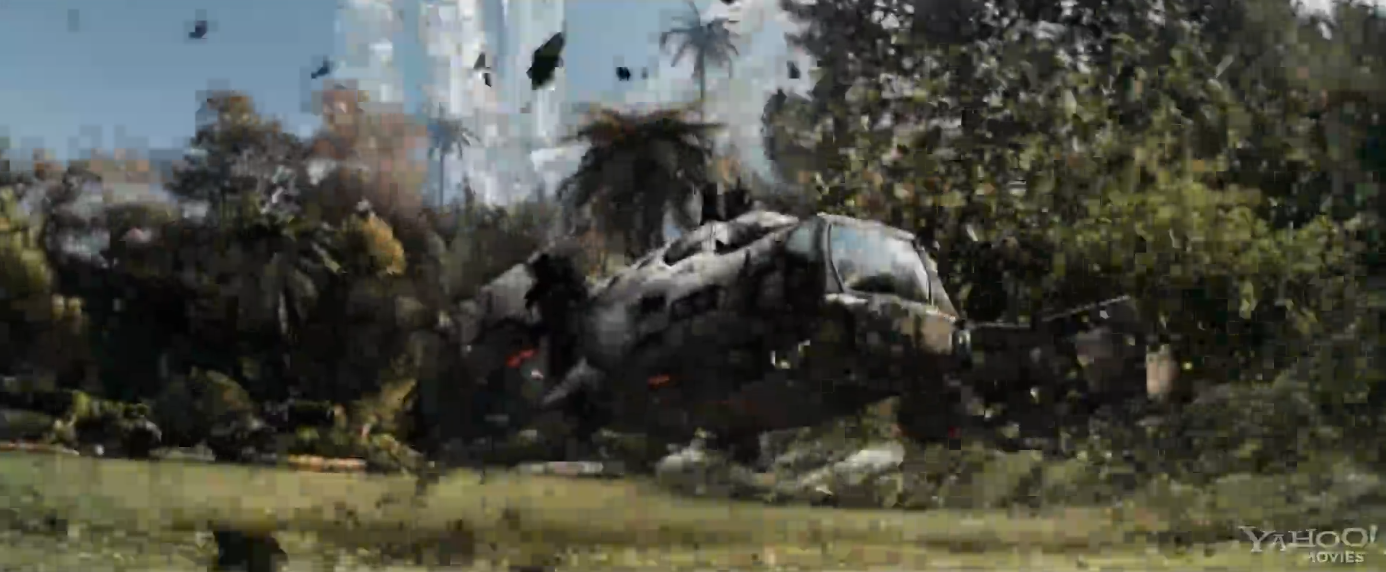} 
\par\end{center}%
\end{minipage}
\par\end{centering}

\begin{centering}
\vspace{0.15in}

\par\end{centering}

\vspace{-0.1in}

\centering{}\caption{Two screenshots from the decoded video of a HAS session with constant
network bandwidth. The video is encoded at multiple CBR levels.}

\label{Flo:2screenshots} \vspace{-0.15in}
 
\end{figure}

The conventional rate adaptation algorithms select the next segment
only based on the video \emph{bitrate} information, but not directly
on the video \emph{quality} of the segment. Thus, during a session,
even if the available network bandwidth stays constant, the delivered
video quality would vary from a high-motion or high-complexity scene
to a low-motion or low-complexity scene. For example, in Figure \ref{Flo:2screenshots},
we show two screenshots from a decoded video of an HAS session streamed
over a constant-bandwidth link. The first screenshot is from the preview
title, which is static and of low complexity. The second one is from
a fairly complex and dynamic scene. Not surprisingly, with CBR encoding
and bitrate-based adaptation, the second screenshot yields much lower
visual quality than the first one.

If we could ``steal'' some bits from the first scene and ``stuff''
them into the second one, the overall viewing experience would have
been better. With this in mind, everything can boil down to an optimization
problem that \emph{temporally} \emph{allocates bits among the video
segments} \emph{to yield an optimal overall quality.} This can be
considered as the basic rationale behind this work.

We define the optimization objective to be an alpha-fairness utility
function \cite{kelly98,srikant04} of some quality metric (for example,
MSE/PSNR \cite{psnr}, Sarnoff/PQR \cite{sarnoff}, MS-SSIM \cite{wang2004image,wang2003multiscale},
VQM \cite{pinson2004new}, STRRED \cite{bovik_strred}), which is
generic enough to cover a range of candidates. Heuristically, this
objective function could balance between total quality and quality
variability.

Besides the bandwidth constraint, the optimization problem is also
subject to two other constraints: 
\begin{itemize}
\item First, the optimization is \emph{myopic} -- it does not know the available
bandwidth in the future,%
\footnote{In this work, we do not attempt to introduce a statistical channel
model and predict the available bandwidth in the future. This allows
our algorithm to be deployed without any assumed knowledge on the
network characteristics. Further improvement can be made if a statistical
channel model is incorporated if appropriate.%
} and furthermore, in light of live streaming, we assume that the client
has visibility into incoming video segments (including both bitrate
and quality information) within a \emph{finite} horizon. 
\item Second, we make use of the client-side video buffer as a breathing
room for video bitrate variability, in a way that the buffer should
neither be completely drained nor fill above a threshold. If the buffer
is completely drained, the playout will stall, which is probably the
worst event for an end user's experience. Typically, to also accommodate
bandwidth variability, the client buffer size should be bounded above
some minimum level (for example, several segments). On the other hand,
due to end-to-end latency in live streaming, or device memory limit,
or simply economic reasons, the buffer size should also be bounded
below some maximum level. 
\end{itemize}
We propose a solution that combines an online algorithm with dynamic
programming. The online algorithm adapts the video bitrate step-by-step,
and at each step we use dynamic programming to solve a constrained
optimization subproblem within a sliding window. The dynamic programming
solution allows us to turn a combinatorial problem into something
solvable in polynomial time. To our advantage is that, in HAS, as
the available bitrate is discrete, it well fits into the dynamic programming
framework.

It is worth noting that, our proposed optimization solution should
reside in a rate adaptation algorithm at the client side. It is fully
orthogonal to the server-side video encoding. For example, in principle\emph{,}
it works with stored video either CBR or variable-bitrate (VBR) encoded
at each level. The only needed architectural change is to convey the
video quality information to the client in some way (e.g., via the
manifest file or an out-of-band approach). 

For a sneak preview of what our algorithm is able to achieve, please
refer to \cite{cqsamples} for some online sample videos.

In the rest of the paper, we first introduce a simple example to illustrate
the intuition (Section \ref{sec:A-Simple-Example}). We then formally
state the problem model and formulate the optimization problem (Section
\ref{sec:Problem-Model}). We derive the dynamic programming solution
for a special case (Section \ref{sec:Dynamic-Programming-Algorithm}),
and use it as a building block for the general online algorithm (Section
\ref{sec:Online-Algorithm}). Then, we present how to incorporate
the optimization solution into PANDA -- a practical client rate adaptation
algorithm designed for large-scale HAS deployment (Section \ref{sec:Incorporation-into-Client}).
Lastly, we present performance evaluation (Section \ref{sec:Performance-Evaluation})
and discuss related work (Section \ref{sec:Related-Work}).

\section{A Simple Example}

\label{sec:A-Simple-Example}

Consider the following simple example. Assume that video content of
$1$ second has already been downloaded and buffered at a client.
The client is now trying to decide which video segment to fetch next.
It has been given the visibility of the video segments of the current
step and one step ahead -- it knows the quality and bitrate information
of their pre-encoded levels. It also has the information of the current
available bandwidth. If assuming that the bandwidth does not change
in the near future, the client can precisely calculate the evolution
of the buffer at the end of each step given that a specific segment
is fetched.

In this example, at the current step, the client is given two choices
-- if downloading the low-quality segment, the buffer gain is $0.5$
second and the resulting segment quality is $1$; if downloading the
high-quality segment, the buffer loss is $0.5$ second and the quality
is $2$. Similarly, at the next step, downloading the low-quality
segment would result in buffer gain of $0.4$ second and segment quality
of $2$, and downloading the high-quality segment would result in
buffer loss of $0.7$ second and segment quality of $4$. Figure \ref{Flo:example}
illustrates all the possible selections and the resulting position
of the client buffer at the end of each step. 

Assuming that at the end of the second step, all that matters is that
the buffer stays above $0$ seconds to avoid video playout stall.
The choice $\{high,high\}$ should not be considered because it results
in negative buffer of $-0.2$ second, implying that the video playout
will stall. Out of the rest possible choices, if the objective is
to maximize the \emph{minimum} quality out of the two segments, the
client should select $\{high,low\}$, yielding best minimum quality
of $2$. $ $As a result, the client should select the $high$-quality
segment to fetch for the current step. On the other hand, if the objective
is to maximize the \emph{total} quality of the two segments, the client
should select $\{low,high\}$, yielding best total quality of 5, and
for the current step, the client should select the $low$-quality
segment to fetch. The same procedure repeats at the next step with
the new bandwidth and video segment information.

\begin{figure}
\begin{centering}
\hspace{-0in} 
\par\end{centering}

\begin{centering}
\begin{minipage}[t]{1\columnwidth}%
\begin{center}
\vspace{-0.22in}

\par\end{center}

\begin{center}
\includegraphics[scale=0.55]{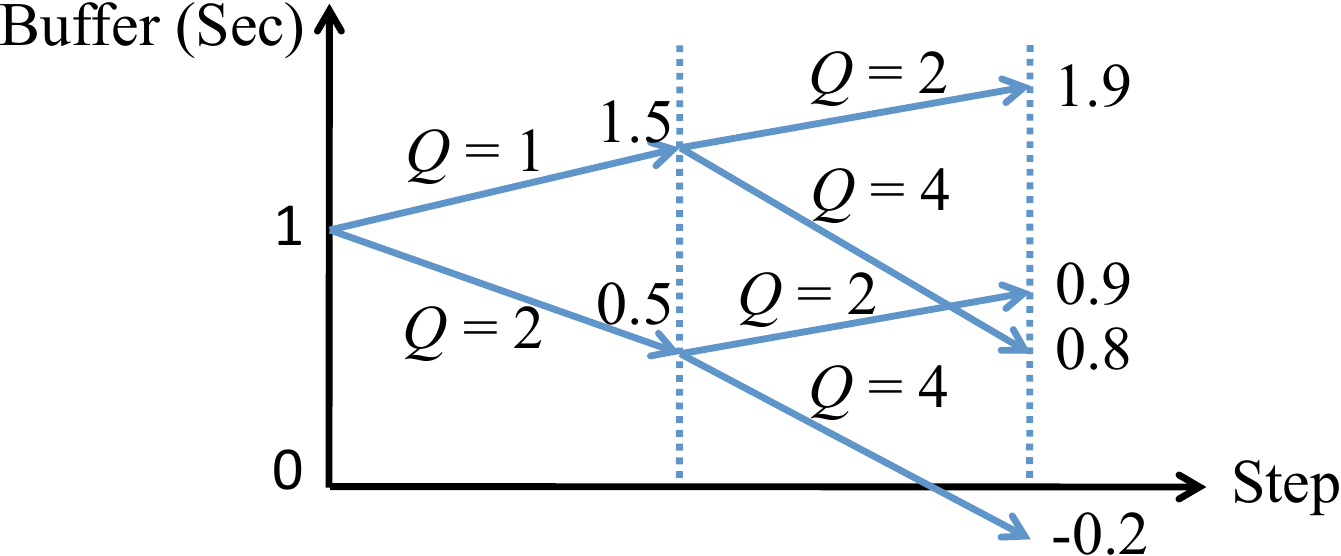} 
\par\end{center}%
\end{minipage}
\par\end{centering}

\centering{}\vspace{-0.05in}
 \caption{A simple example. A client selects current segment to be fetched,
based on information of available bandwidth, buffer size and knowledge
of pre-encoded video within a finite horizon of two segments.}

\label{Flo:example} \vspace{-0.1in}
 
\end{figure}

\section{Problem Model}

\label{sec:Problem-Model}

A video stream is chopped into segments $n=1,2,3,...$ of $\tau$
seconds. Each segment has been pre-encoded at $L$ levels. The resulting
quality and bitrate of the $n$-th segment of the $l$-th level are
denoted by $Q(n,l)$ and $R(n,l)$, respectively. In conventional
CBR encoding, it is assumed that $R(n,l)=R(m,l)$ for $n\neq m$.
Our problem model assumes the most general case where $R(n,l)$ at
level $l$ can vary from segment to segment.

At each adaptation step $n$, the client selects bitrate $R(n)$ for
the current segment to be fetched, from a finite set of available
bitrates $\{R(n,l):l=1,2,...,L\}$. The quality corresponding to the
selection $R(n)$ is denoted by $Q(n)$.

\subsection{Objectives}

To define the quality optimization objective, we introduce the notion
of $\alpha$-fairness \cite{kelly98,srikant04}. For a quality value
$q$, its $\alpha$-fairness utility is defined as 
\[
U_{\alpha}(q):=\frac{q^{1-\alpha}}{1-\alpha}.
\]
Our optimization objective is to maximize the sum of the $\alpha$-fairness
utility over a set of segments: 
\begin{equation}
\max\sum_{n}U_{\alpha}(Q(n)).\label{eq:obj1}
\end{equation}
The rationale is to model resource allocation among a set of segments
as a fairness problem. This objective function encompasses a number
of special cases. At one extreme, setting $\alpha=0$ corresponds
to utilitarianism that accounts for total quality (i.e., ``max-sum'').
At the other extreme, setting $\alpha\rightarrow\infty$ corresponds
to max-min fairness and minimum quality variability. Setting $\alpha$
between 0 and $\infty$ achieves balance between total quality and
quality variability. For example, $\alpha=1$ results in the limiting
form $U_{1}(q)=\log q$, and this corresponds to proportional fairness
and the Nash bargaining solution \cite{kelly98}.

An extension to (\ref{eq:obj1}) is to more precisely consider the
quality variation by incorporating a multiplier $\delta(n-1,n)$ at
each step $n$, and use the following objective function: 
\begin{equation}
\max\sum_{n}\delta(n-1,n)\cdot U_{\alpha}(Q(n)).\label{eq:obj2}
\end{equation}
The multiplier $\delta(n-1,n)$ discounts the overall utility if the
quality has shifted from one level to another. For example, set $\delta(n-1,n)=1$
if the segments at step $n-1$ and $n$ are selected from the same
level, and set $\delta(n-1,n)=0.9$ if they are from different levels.
Note that, (\ref{eq:obj2}) would be useful if the video source is
encoded such that each level corresponds to a constant quality, i.e.,
$Q(n,l)=Q(m,l)$ for $n\neq m$.

In the following discussions, we assume that the objective function
follows the general form $\max\sum_{n}U(n)$ where the utility function
$U(n)$ can be either $U_{\alpha}(Q(n))$ or $\delta(n-1,n)\cdot U_{\alpha}(Q(n))$.

\subsection{Constraints}

Let $B(0)$ be the initial buffer size (measured in content seconds),
and $B(n)$ the buffer size at the end\emph{ }of step $n$. After
video playout starts, the buffer evolution can be modeled as 
\begin{equation}
B(n)=B(n-1)+\tau-\tau\cdot R(n)/W(n)\label{eq:bn2}
\end{equation}
where $W(n)$ is the link bandwidth at step $n$, and $\tau\cdot R(n)/W(n)$
is the segment download duration. That is, in each step, the replenishment
of the buffer is $\tau$ seconds, and the depletion of the buffer
is $\tau\cdot R(n)/W(n)$ seconds.

The optimization must be subjected to the constraint of client buffer
size. Define $B_{L}$ and $B_{H}$ to be the lower and upper buffer
bound, respectively, with $0\leq B_{L}\leq B_{H}$. Except for the
initial state where $B(n)<B_{L}$, or for when there is sudden bandwidth
variation, the buffer should be maintained such that $B_{L}\leq B(n)\leq B_{H}$.
Furthermore, we define a buffer reference level $B_{0}$, towards
which the buffer level attempts to converge to. 

We note that setting the lower bound $B_{L}$ achieves the balance
between the video variability and the bandwidth variability that can
be compensated -- the higher the $B_{L}$, the more bandwidth variability
that can be accommodated, but the less breathing room for video quality
variability; vice versa.

Another constraint is the finite horizon $H(n)$, $H(n)\geq1$, in
number of adaptation steps (or segments). That is, at step $n$, only
the bitrate and quality information of the segments from now to $H(n)-1$
steps ahead, i.e., $\{(R(m,l),Q(m,l)):m=n,n+1,...,n+H(n)-1,l=1,2,...,L\}$,
are available to the client's optimization algorithm.

\section{Dynamic Programming}

\label{sec:Dynamic-Programming-Algorithm}

In this section, we describe the proposed dynamic programming solution,
which is used as a building block in each step of the online algorithm.
Within a finite horizon of $H$ steps, given an initial buffer size
$B_{init}\in[B_{L},B_{H}]$ and a final buffer size $B_{final}\in[B_{L},B_{H}]$,
the dynamic programming algorithm attempts to solve the following
problem: 
\[
\begin{array}{rcl}
 & \max_{\{R(n)\}} & \sum_{n=1}^{H}U(n),\\
 & \textrm{s.t.} & \sum_{n=1}^{H}R(n)\leq H\cdot W,\\
 &  & B(0)=B_{init},\\
 &  & B(H)=B_{final},\\
 &  & B_{L}\leq B(n)\leq B_{H}\textrm{ for }n=1,...,H,
\end{array}
\]
where $W$ is the available bandwidth, assumed to be constant. The
specific value of $B_{final}$ used will be set in the online algorithm
and passed as an input to the dynamic programming algorithm. 

How can we solve this problem within polynomial time? Our dynamic
programming solution is based on the following intuition. Recall that
$R(n)$ is the segment bitrate selected at step $n$. $R(n)$ has
a corresponding utility $U(n)$, which can be evaluated based on (\ref{eq:obj1})
or (\ref{eq:obj2}). Alternatively, we can write $U(n)$ in terms
of the buffer evolution 
\begin{equation}
U((n-1,B(n-1))\rightarrow(n,B(n)))=U(n),\label{eq:utility1}
\end{equation}
where the buffer evolves according to (\ref{eq:bn2}). The interpretation
of (\ref{eq:utility1}) is that, by moving the buffer from position
$B(n-1)$ at step $n-1$ to $B(n)$ at step $n$, the utility is $U(n)$.
Similar utility $U((m,B(m))\rightarrow(n,B(n)))$ can be defined for
the case of $n-m\geq2$. However, notice that in this case, there
may be multiple possible paths to move the buffer from $(m,B(m))$
to $(n,B(n))$, which may result in different utility values. We can
then define $U^{*}((m,B(m))\rightarrow(n,B(n)))$ to be the maximum
utility over all the possible paths.

The key to the dynamic programming solution is to realize that the
solution to the problem of $H$ steps can be formed using solutions
to the problems of lesser steps. First, we show that the following
theorem is true: \newtheorem{theorem}{Theorem} \begin{theorem}
Let $s$ be any intermediate step between two non-adjacent steps $m$
and $n$, or $m<s<n$. It holds true that 
\begin{eqnarray*}
 &  & U^{*}((m,B(m))\rightarrow(n,B(n)))\\
 & = & \max_{B(s)\in[B_{L},B_{H}]}U^{*}((m,B(m))\rightarrow(s,B(s)))\textrm{ }\\
 &  & \textrm{ }\textrm{ }\textrm{ }\textrm{ }\textrm{ }\textrm{ }\textrm{ }\textrm{ }\textrm{ }\textrm{ }\textrm{ }\textrm{ }\textrm{ }\textrm{ }\textrm{ }\textrm{ }\textrm{ }\textrm{ }\textrm{ }\textrm{ }\textrm{ }\textrm{ }\textrm{ }+\textrm{ }U^{*}((s,B(s))\rightarrow(n,B(n))).
\end{eqnarray*}
\end{theorem} \begin{proof} Let $B^{*}(s)\in[B_{L},B_{H}]$ be the
buffer size at step $s$ that maximizes $U^{*}((m,B(m))\rightarrow(s,B(s)))+\textrm{ }U^{*}((s,B(s))\rightarrow(n,B(n)))$.
For any path from $(m,B(m))$ to $(n,B(n))$, let $B^{\star}(s)\in[B_{L},B_{H}]$
be the buffer at $s$ that the path has passed. It holds true that
\begin{eqnarray}
 &  & U((m,B(m))\rightarrow(n,B(n)))\nonumber \\
 & = & U((m,B(m))\rightarrow(s,B^{\star}(s)))\nonumber \\
 &  & \textrm{ }\textrm{ }\textrm{ }\textrm{ }\textrm{ }\textrm{ }\textrm{ }\textrm{ }\textrm{ }\textrm{ }+U((s,B^{\star}(s))\rightarrow(n,B(n)))\nonumber \\
 & \leq & U^{*}((m,B(m))\rightarrow(s,B^{\star}(s)))\nonumber \\
 &  & \textrm{ }\textrm{ }\textrm{ }\textrm{ }\textrm{ }\textrm{ }\textrm{ }\textrm{ }\textrm{ }\textrm{ }+U^{*}((s,B^{\star}(s))\rightarrow(n,B(n)))\label{eq:proof1}\\
 & \leq & U^{*}((m,B(m))\rightarrow(s,B^{*}(s)))\nonumber \\
 &  & \textrm{ }\textrm{ }\textrm{ }\textrm{ }\textrm{ }\textrm{ }\textrm{ }\textrm{ }\textrm{ }\textrm{ }+U^{*}((s,B^{*}(s))\rightarrow(n,B(n)))\label{eq:proof2}
\end{eqnarray}
where (\ref{eq:proof1}) is by the definition of $U^{*}$ and (\ref{eq:proof2})
is by the definition of $B^{*}(s)$. The optimal value is achievable
by selecting $B^{\star}(s)=B^{*}(s)$ and recursively selecting the
optimal sub-paths. \end{proof} This theorem states that, going from
$(m,B(m))$ to $(n,B(n))$, inevitably one has to pass a mid-way step
$s$. At step $s$, one could have many possible buffer sizes $B(s)$.
It holds true that, the optimal utility value of a problem $(m,B(m))\rightarrow(n,B(n))$
has to be the sum of the optimal utility values of the sub-problems
$(m,B(m))\rightarrow(s,B(s))$ and $(s,B(s))\rightarrow(n,B(n))$
over all possible $B(s)$. So one can solve the problem by solving
its sub-problems, by solving its sub-sub-problems, and so on. Eventually,
things reduce to the baseline case of $(m-1,B(m-1))\rightarrow(m,B(m))$.
Once a sub-problem has been solved, one can store the solution (including
the optimal utility values and some side information for backtracking
purpose) in a table for later reuse to save repeated work.

\begin{algorithm}
\small \textbf{Input}:\vspace{-0.05in}
\begin{itemize}
\item $B_{init}$, $B_{final}$, $B_{L}$, $B_{H}$, $\tau$, $W$, $H$\vspace{-0.05in}
\item $\{(R(m,l),Q(m,l)):m=1,...,H,l=1,2,...,L\}$\vspace{-0.05in}
\end{itemize}
\textbf{Output}:\vspace{-0.05in}
\begin{itemize}
\item $\{R(m):m=1,...,H\}$ \vspace{-0.05in}
\end{itemize}
\textbf{Procedure}:\vspace{-0.05in}
\begin{itemize}
\item Let bin $k$ corresponds to $B_{init}\in I_{k}$. Store in the table
$U^{*}(0,k)=0$ and $B^{*}(0,k)=B_{init}$. $ $\vspace{-0.05in}
\item For step $m=1,2,...,H$: \vspace{-0.05in}

\begin{itemize}
\item For bin $k=1,...,K$, if $U^{*}(m-1,k)$ already has value stored:\vspace{-0.05in}

\begin{itemize}
\item For level $l=1,...,L$: \vspace{-0.05in}

\begin{itemize}
\item Calculate $U=U^{*}(m-1,k)+U((m-1,B^{*}(m-1,k))\rightarrow(m,B(m)))$,
where the second term $U((m-1,B^{*}(m-1,k))\rightarrow(m,B(m)))$
corresponds to fetching $R(m,l)$. Record $B(m)$.\vspace{-0.05in}
\item Let bin $k'$ corresponds to $B(m)\in I_{k'}$. If $U^{*}(m,k')$
has no value stored yet or the currently stored value $U^{*}(m,k')<U$,
set $U^{*}(m,k')=U$ and store the corresponding side information
$B^{*}(m,k')=B(m)$.\vspace{-0.05in}
\end{itemize}
\end{itemize}
\end{itemize}
\item Backtrack to get the optimal bitrates $\{R^{*}(m):m=1,...,H\}$ that
yields $U^{*}(H,k)$, where bin $k$ corresponds to $B_{final}\in I_{k}$. \vspace{-0.05in}
\item Output $\{R^{*}(m):m=1,...,H\}$.\vspace{-0.05in}
\end{itemize}
\caption{Dynamic Programming \label{alg:DP}}
\end{algorithm}

An implementation detail is that, as $[B_{L},B_{H}]$ is a continuous
interval but the dynamic programming is discrete, we need to quantize
$[B_{L},B_{H}]$ into $K$ discrete bins with step size $\Delta B$,
and only store one optimal utility value for each bin. Denote by $\mathcal{I}:=\{I_{1},I_{2},...,I_{K}\}$
the resulting bins. Thus, the optimal utility values $U^{*}(n,k)$
can be stored in an $(H+1)\times K$ two-dimensional table where the
first dimension corresponds to the number of steps (including the
initial zeroth step) and the second corresponds to the bins. In the
table, we also store the side information \textbf{$B^{*}(n,k)$},
which is the ending buffer size corresponding to $U^{*}(n,k)$.

The dynamic programming solution is described in Algorithm \ref{alg:DP}.
A simple analysis shows that the algorithm has complexity $O(H\cdot K\cdot L)$.
In practical implementation in C, we find that with typical parameters
(e.g., $H=30$, $K=50$, $L=10$), the execution time is within a
few milli-seconds (e.g., 5 ms).

Note that one corner case is, in the second last step of Algorithm
\ref{alg:DP}, we may not find a bin $k$ such that $B_{final}\in I_{k}$
and there is a value $U^{*}(H,k)$ in it. This may happen if the available
bandwidth is either too large or too small for the available pre-encoded
video bitrates (recall that we assume no off-intervals between segment
downloading in the dynamic programming problem formulation). If this
happens, we may find another bin $k''$ which has a value $U^{*}(H,k'')$
stored and is closest to $k$, and then perform the backtrace starting
from $k''$. In this case, we also output the buffer offset 
\begin{equation}
B_{offset}=B^{*}(H,k'')-B^{*}(H,k)\label{eq:buffer_offset}
\end{equation}
to be later used in Section \ref{sec:Incorporation-into-Client}.
\vspace{-0.05in}

\section{Online Algorithm}

\label{sec:Online-Algorithm}

If the available bandwidth does not vary and the video quality information
is available all at once, the dynamic programming algorithm is sufficient
to solve the optimization problem in one shot. In reality, the bandwidth
changes over time and the video quality information is available within
a finite horizon. To deal with this, we propose an online algorithm
that repeatedly applies the dynamic programming in a sliding-window
manner.

We define a reference buffer level $B_{0}$ that the buffer aims to
converge to. We also have a buffer lower bound $B_{L}$ and a buffer
upper bound $B_{H}$. However, different from the problem formulation
in Section \ref{sec:Dynamic-Programming-Algorithm}, we do not guarantee
that the buffer at transient state is bounded within $[B_{L},B_{H}]$.
For example, when the streaming starts, the initial buffer is zero.
Instead, at a particular step $n$, we only make sure that the buffer
is bounded within $[\min(B_{L},B(n-1)),\max(B_{H},B(n-1))]$. But
over time, thanks to the convergence to $B_{0}$, the buffer is set
to be bounded within $[B_{L},B_{H}]$ if the bandwidth does not abruptly
change.

The finite horizon size at step $n$ is denoted by $H(n)$, which
may vary over time in some applications. For example, in live streaming,
as the end-to-end latency is bounded, the longer the buffered video
is at the client, the shorter the horizon is. Within a window of size
$H(n)$, the dynamic programming algorithm is applied, and an optimal
rate allocation $\{R^{*}(m):m=1,...,H(n)\}$ is obtained. Then the
bitrate of the current segment to be fetched is set to be $R^{*}(1)$,
i.e., only the most immediate rate is applied.

The online algorithm is described in Algorithm \ref{alg:online}.

\begin{algorithm}
\small \textbf{Global input}:\vspace{-0.05in}
\begin{itemize}
\item $B_{L}$, $B_{H}$, $B_{0}$, $\tau$\vspace{-0.05in}
\end{itemize}
\textbf{Input at step $n$}:\vspace{-0.05in}
\begin{itemize}
\item $W(n)$, $B(n-1)$, $H(n)$ \vspace{-0.05in}
\item $\{(R(m,l),Q(m,l)):m=n,...,n+H(n)-1,l=1,...,L\}$ \vspace{-0.05in}
\end{itemize}
\textbf{Output at step }$n$:\vspace{-0.05in}
\begin{itemize}
\item $R(n)$\vspace{-0.05in}
\end{itemize}
\textbf{Procedure at step $n$}:\vspace{-0.05in}
\begin{itemize}
\item $B_{init}=B(n-1)$.\vspace{-0.05in}
\item $B_{final}=B_{0}$.\vspace{-0.05in}
\item $B_{L}(n)=\min(B_{L},B(n-1))$.\vspace{-0.05in}
\item $B_{H}(n)=\max(B_{H},B(n-1))$.\vspace{-0.05in}
\item $\{R^{*}(m):m=1,...,H(n)\}$
\vspace{-0.07in}
\begin{eqnarray*}
 & = & \mathrm{\textrm{DynamicProgramming}}(B_{init},B_{final},B_{L}(n),\\\vspace{-0.05in}
 &  & B_{H}(n),\tau,W(n),H(n),\{(R(m,l),Q(m,l)):\\\vspace{-0.05in}
 &  & m=n,...,n+H(n)-1,l=1,...,L\}).\vspace{-0.05in}
\end{eqnarray*}\vspace{-0.2in}
\item Output $R(n)=R^{*}(1)$. \vspace{-0.05in}
\end{itemize}
\caption{Online Algorithm\label{alg:online}}
\end{algorithm}

\section{PANDA with Consistent Quality}

\label{sec:Incorporation-into-Client}

PANDA (reading: Probe-AND-Adapt) is an HAS client rate adaptation
algorithm we recently designed to yield high stability and fast responsiveness
to bandwidth variations when multiple HAS clients are running within
a network domain sharing bottleneck links \cite{panda13}. Performance
evaluations show that, compared to conventional HAS algorithms, PANDA
is able to reduce the instability of video bitrate selection by over
75\% without increasing the risk of buffer underrun. To detect the
available bandwidth, PANDA \emph{probes} the network by additively
incrementing its sending rate at each adaptation step and multiplicatively
decreasing its rate if congestion is detected, and \emph{adapts} its
video bitrate accordingly. This ``probe and adapt'' principle is
akin to the additive increase / multiplicative decrease (AIMD) principle
used in TCP, but it operates in the application layer and at a much
longer time scale.

The original PANDA design is video quality-agnostic. In this section,
we extend it to incorporate video quality optimization. Our online
algorithm naturally fits into PANDA, with the probing part of PANDA
providing the bandwidth estimation for the online algorithm, while
the online algorithm determining the next segment to fetch and the
target inter-request time. We name the new algorithm \emph{PANDA with
Consistent Quality (PANDA/CQ)}.

\begin{algorithm}
\small \textbf{Global input}:\vspace{-0.05in}
\begin{itemize}
\item $\kappa$, $w$, $a$, $B_{L}$, $B_{H}$, $B_{0}$, $\tau$, $\beta$\vspace{-0.05in}
\end{itemize}
\textbf{Input at step $n$}:\vspace{-0.05in}
\begin{itemize}
\item $B(n-1)$, $H(n)$, $T(n-1)$, $\tilde{x}(n-1)$\vspace{-0.05in}
\item $\{(R(m,l),Q(m,l)):m=n,...,n+H(n)-1,l=1,...,L\}$ \vspace{-0.05in}
\end{itemize}
\textbf{Output at step }$n$:\vspace{-0.05in}
\begin{itemize}
\item $R(n)$, $\hat{T}(n)$\vspace{-0.05in}
\end{itemize}
\textbf{Procedure at step $n$}:\vspace{-0.05in}
\begin{itemize}
\item Estimate the bandwidth share $\hat{x}(n)$ by solving:\vspace{-0.05in}
\[
\frac{\hat{x}(n)-\hat{x}(n-1)}{T(n-1)}=\kappa\cdot(w-\max(0,\hat{x}(n-1)-\tilde{x}(n-1)+w)).\vspace{-0.05in}
\]

\item Smooth out $\hat{x}(n)$ to produce its filtered version $\hat{y}(n)$\vspace{-0.05in}
by solving:
\[
\frac{\hat{y}(n)-\hat{y}(n-1)}{T(n-1)}=-a\cdot(\hat{y}(n-1)-\hat{x}(n)).\vspace{-0.05in}
\]

\item Apply the online algorithm to pick the fetched video bitrate $R(n)$:\vspace{-0.05in}
\begin{eqnarray*}
R(n) & = & \textrm{Online}(B_{L},B_{H},B_{0},\tau,\hat{y}(n),B(n-1),H(n),\\
 &  & \{(R(m,l),Q(m,l)):m=n,...,n+H(n)-1,\\
 &  & l=1,...,L\}).
\end{eqnarray*}
\vspace{-0.2in}
\item Determine the target time until the next request $\hat{T}(n)$ by:\vspace{-0.05in}
\begin{equation}
\hat{T}(n)=\frac{R(n)\cdot\tau}{\hat{y}(n)}+\beta\cdot(B(n-1)-B_{0}).\vspace{-0.05in}\label{eq:qopanda_step4}
\end{equation}

\end{itemize}
\caption{PANDA/CQ\label{alg:qo-panda}}
\end{algorithm}

\begin{center}
\begin{table}[t]
{\scriptsize }%
\begin{minipage}[t]{0.99\columnwidth}%
\begin{center}
\begin{tabular}{|l|c|c|}
\hline 
Algorithm  & Parameter  & Default Value\tabularnewline
\hline 
PANDA and PANDA/CQ  & $\kappa$  & 0.28\tabularnewline
 & $w$  & 0.3\tabularnewline
 & $a$  & 0.2\tabularnewline
 & $\beta$  & 0.2\tabularnewline
 & $\tau$  & 2\tabularnewline
\hline 
PANDA only  & $B_{0}$  & 20\tabularnewline
 & $\epsilon$  & 0\tabularnewline
\hline 
PANDA/CQ only  & $B_{0}$  & 30\tabularnewline
 & $B_{L}$  & 10\tabularnewline
 & $B_{H}$  & 50\tabularnewline
 & $H$  & 30\tabularnewline
\hline 
\end{tabular}
\par\end{center}%
\end{minipage}{\scriptsize \caption{{\scriptsize \label{tab:parameters} }Default client parameters in
ns-2 simulations}
}{\scriptsize \par}

{\scriptsize \vspace{-0.1in}
 }
\end{table}

\par\end{center}

\begin{figure}
\begin{centering}
\includegraphics[scale=0.6]{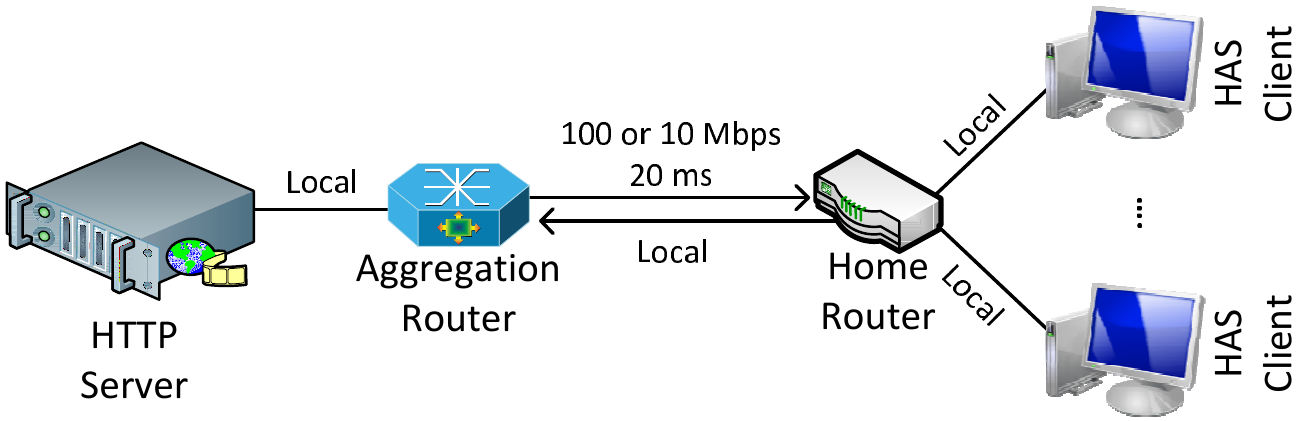} \vspace{-0.05in}
 
\par\end{centering}

\centering{}\caption{The network topology configured in the ns-2 simulator. Local indicates
that the bitrate is effectively unbounded and the link delay is 0
ms.}

\label{Flo:testbed_network} \vspace{-0.1in}
 
\end{figure}

The PANDA/CQ is described in Algorithm \ref{alg:qo-panda}. For each
adaptation step $n$, it proceeds in four sub-steps. In sub-step one,
it estimates the available bandwidth using probing, with $k$ and
$w$ the two probing parameters controlling the probing convergence
rate and the multiplicative decrease threshold, respectively. $T(n)$
is the time duration of step $n$, equal to $\max(\hat{T}(n),\tilde{T}(n))$,
where $\tilde{T}(n)$ is the duration for downloading segment $n$.
$\tilde{x}(n)$ is the calculated TCP throughput, based on formula
$\tilde{x}(n)=R(n)\cdot\tau/\tilde{T}(n)$. The resulting rate $\hat{x}(n)$
is the raw estimation of the bandwidth. In sub-step two, it smoothes
out the raw estimation via exponentially weighted moving average (EWMA)
filtering, to produce the filtered version of the bandwidth estimation,
$\hat{y}(n)$. Here $a$ is a parameter controlling the filtering
convergence rate. In sub-step three, $\hat{y}(n)$ is taken as the
input available bandwidth (i.e., $W(n)$) of the online algorithm,
which generates a video bitrate $R(n)$, to be fetched in the current
step. In the last sub-step, it calculates the target inter-fetch interval
$\hat{T}(n)$ based on $R(n)$ and $\hat{y}(n)$. $\hat{T}(n)$ also
compensates for the current buffer offset $B(n-1)-B_{0}$, with parameter
$\beta>0$ controling the convergence speed.

Note that the current form of Algorithm \ref{alg:qo-panda} does not
handle the corner case of available bandwidth being too large for
the available pre-encoded video bitrates. Recall that in Section \ref{sec:Dynamic-Programming-Algorithm},
we discuss that there may be an offset between the target final buffer
size $B_{final}$ and the actual ending buffer size the algorithm
produces, assuming no off-intervals between segments. If this offset
is positive, we can introduce off-intervals between segment downloads
to compensate for the offset (if negative, there is nothing we can
do). Thus, to handle this corner case, we can simply replace (\ref{eq:qopanda_step4})
in the last sub-step with 
\[
\hat{T}(n)=\frac{R(n)\cdot\tau}{\hat{y}(n)}+\beta\cdot(B(n-1)-B_{0})+\frac{\max(B_{offset}(n),0)}{H(n)}
\]
where $B_{offset}(n)$ is calculated according to (\ref{eq:buffer_offset}).

\begin{figure*}
\begin{centering}
\begin{minipage}[t]{0.8\columnwidth}%
\begin{center}
\includegraphics[scale=0.35]{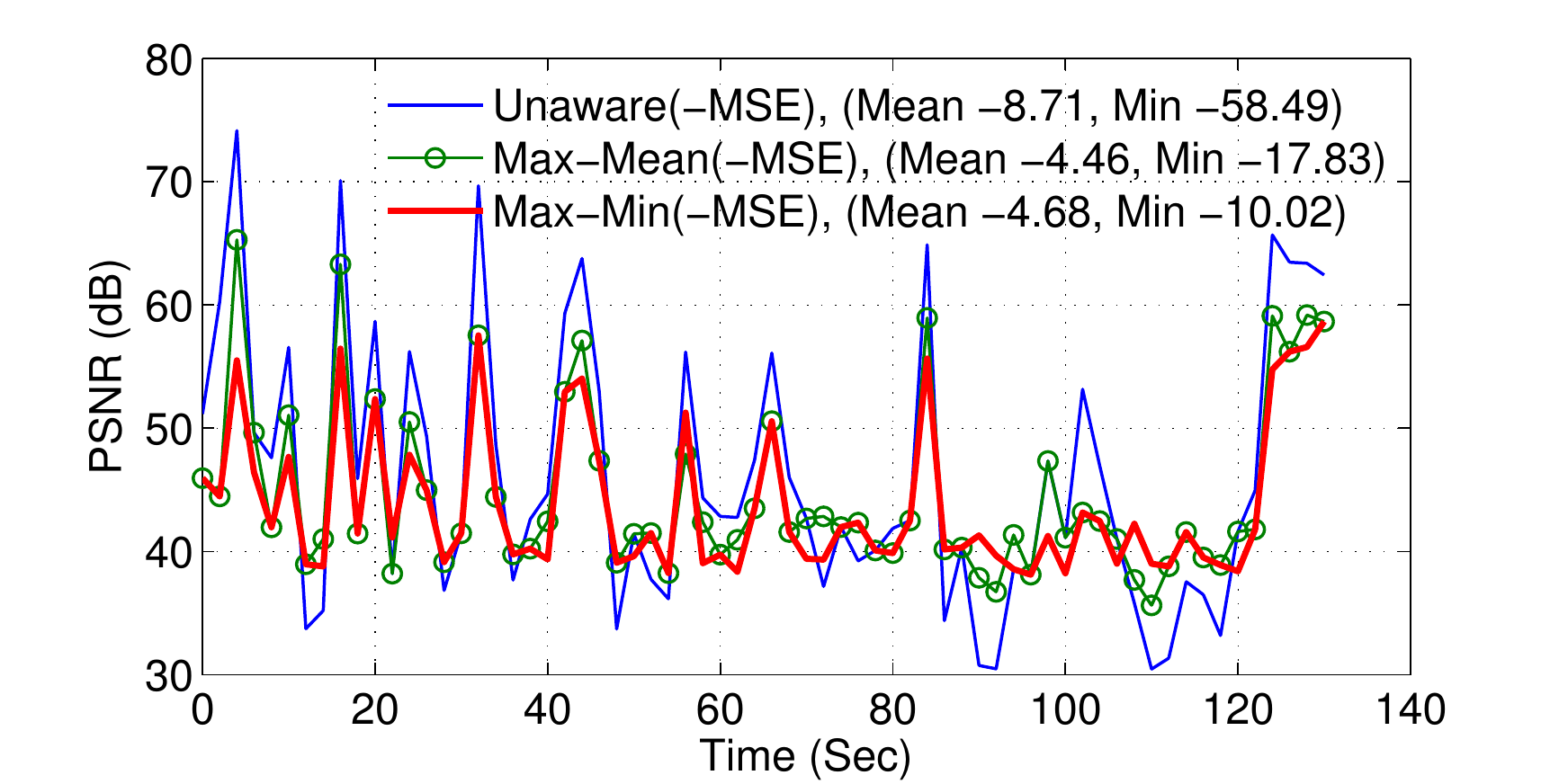} 
\par\end{center}

\begin{center}
\vspace{-0.05in}
 {\footnotesize (a1) Quality (}\emph{\footnotesize Elysium)}{\footnotesize{} }
\par\end{center}{\footnotesize \par}

\begin{center}
{\footnotesize \vspace{-0.25in}
}
\par\end{center}{\footnotesize \par}

\begin{center}
{\footnotesize \includegraphics[scale=0.35]{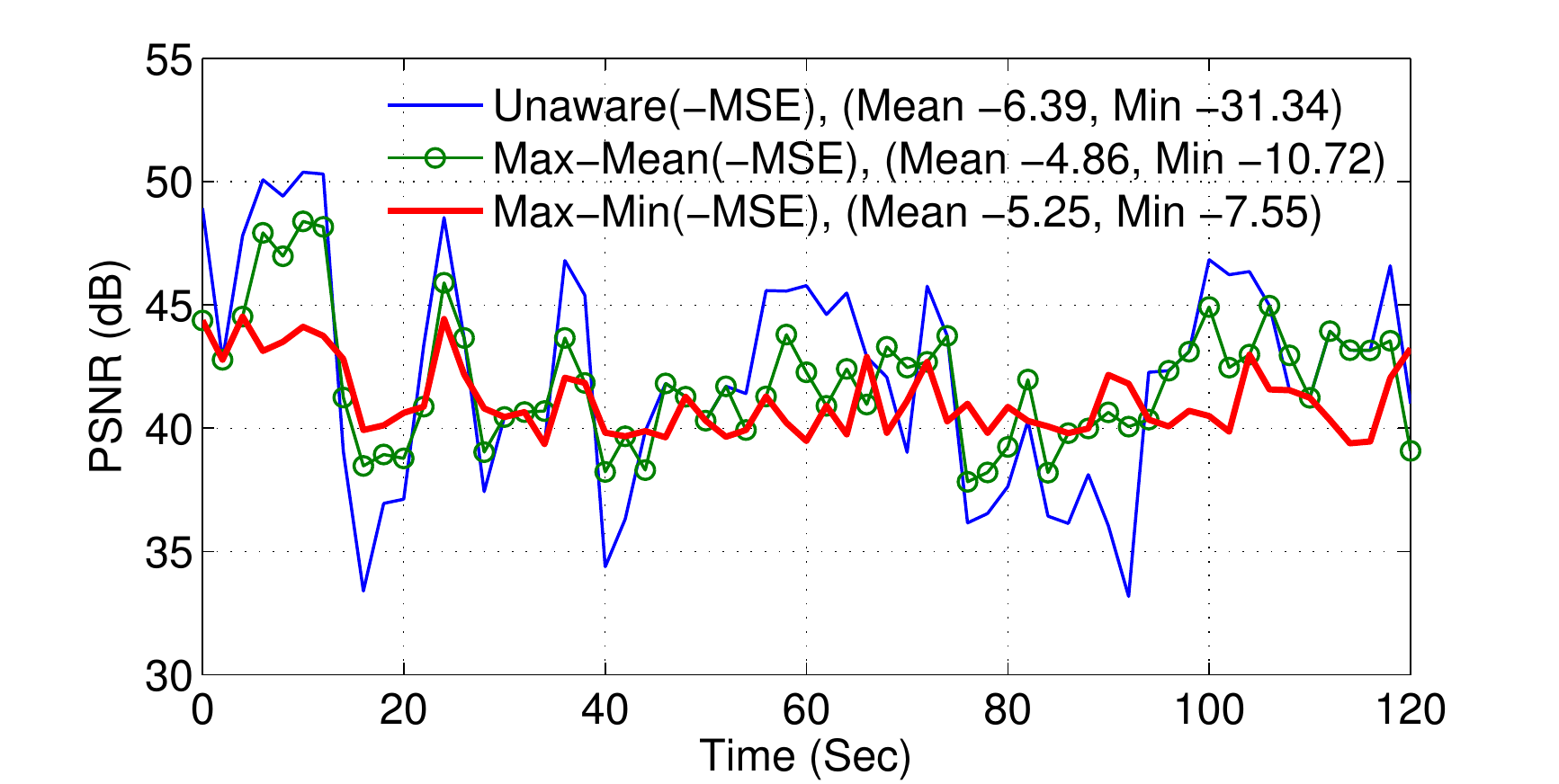} }
\par\end{center}{\footnotesize \par}

\begin{center}
{\footnotesize \vspace{-0.05in}
 (b1) Quality (}\emph{\footnotesize Avatar}{\footnotesize ) }
\par\end{center}%
\end{minipage}%
\begin{minipage}[t]{0.8\columnwidth}%
\begin{center}
\includegraphics[scale=0.35]{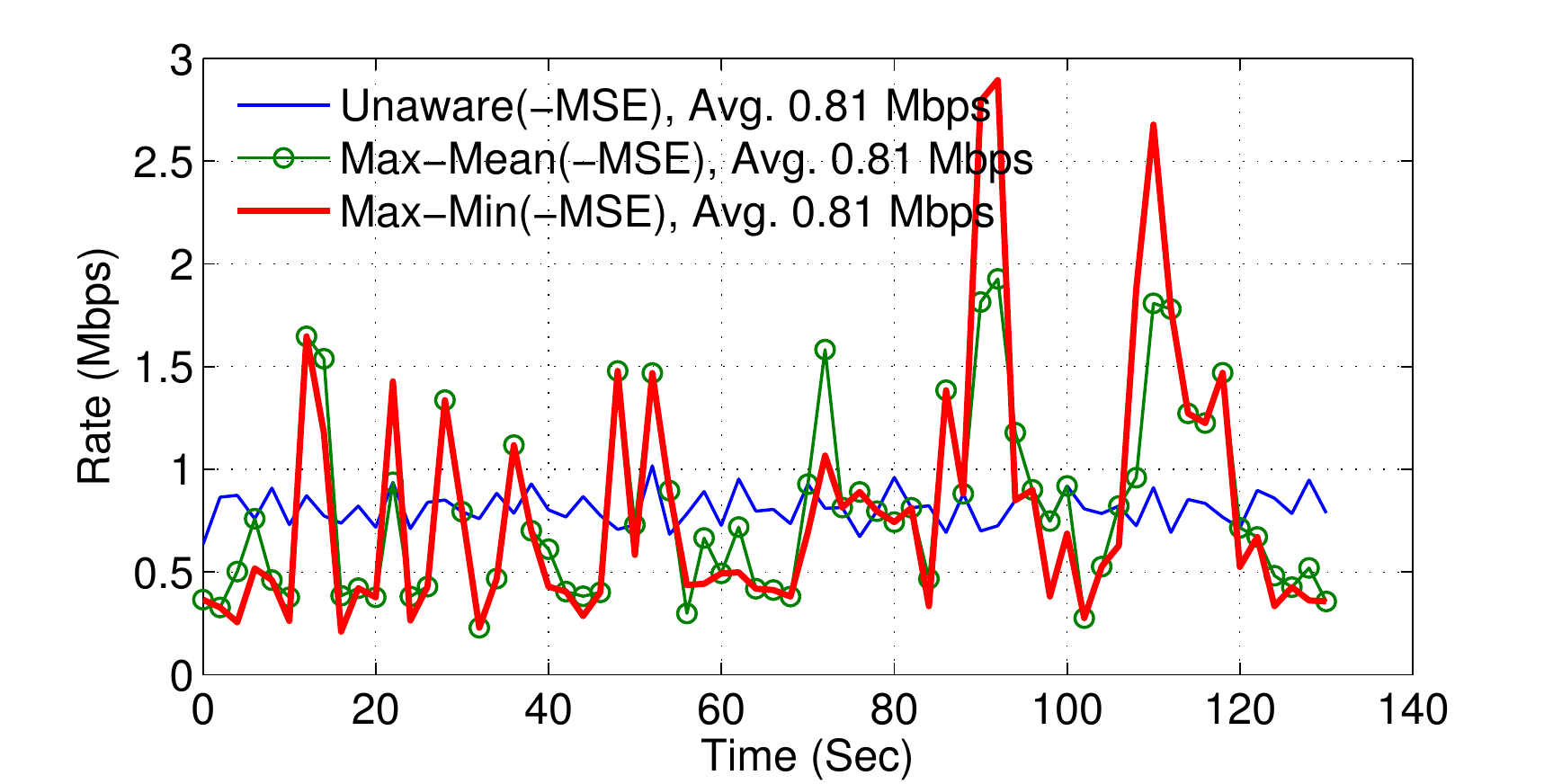} 
\par\end{center}

\begin{center}
\vspace{-0.05in}
 {\footnotesize (a2) Rate (}\emph{\footnotesize Elysium}{\footnotesize ) }
\par\end{center}{\footnotesize \par}

\begin{center}
{\footnotesize \vspace{-0.25in}
}
\par\end{center}{\footnotesize \par}

\begin{center}
{\footnotesize \includegraphics[scale=0.35]{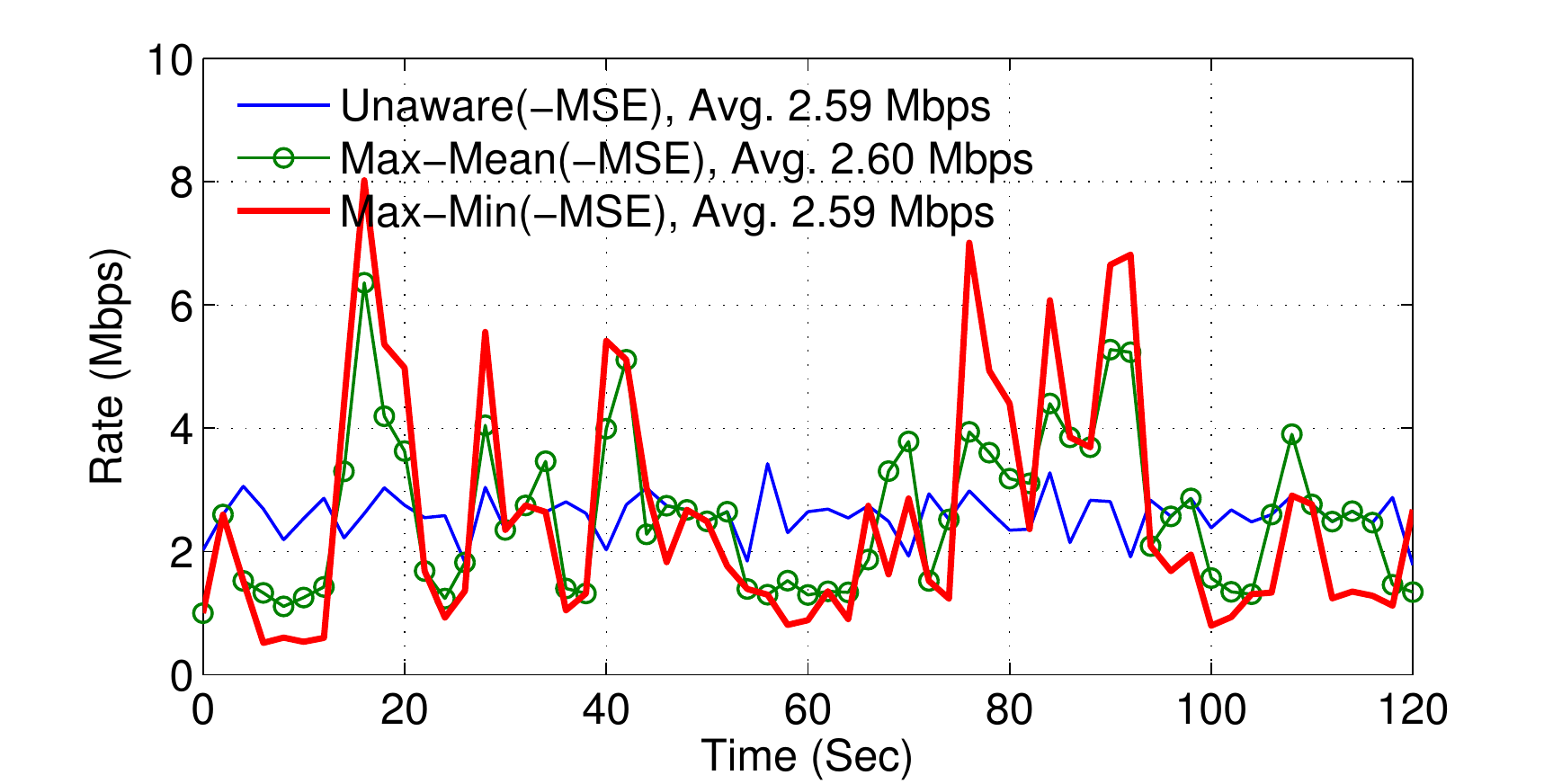} }
\par\end{center}{\footnotesize \par}

\begin{center}
{\footnotesize \vspace{-0.05in}
 (b2) Rate (}\emph{\footnotesize Avatar}{\footnotesize ) }
\par\end{center}%
\end{minipage}
\par\end{centering}

\vspace{0.05in}

\caption{Comparing the traces of three schemes: 1) bitrate-based fetching that
is unaware of the quality information (Unaware(-MSE)), 2) dynamic
programming solution that maximizes the minimal quality (Max-Min(-MSE))
and 3) dynamic programming solution that maximizes the mean quality
(Max-Mean(-MSE)). The available bandwidth is set at constant 0.81
Mbps for \emph{Elysium} and 2.60 Mbps for \emph{Avatar}. The initial
and final buffer levels are 30 seconds; the buffer lower and upper
bounds are 20 and 50 seconds, respectively. The reported quality in
-MSE is converted to PSNR using (\ref{eq:psnr}) for better display.}

\label{Flo:exp_optobj} \vspace{-0.1in}
 
\end{figure*}

\begin{figure*}
\begin{centering}
\begin{minipage}[t]{0.8\columnwidth}%
\begin{center}
\includegraphics[scale=0.35]{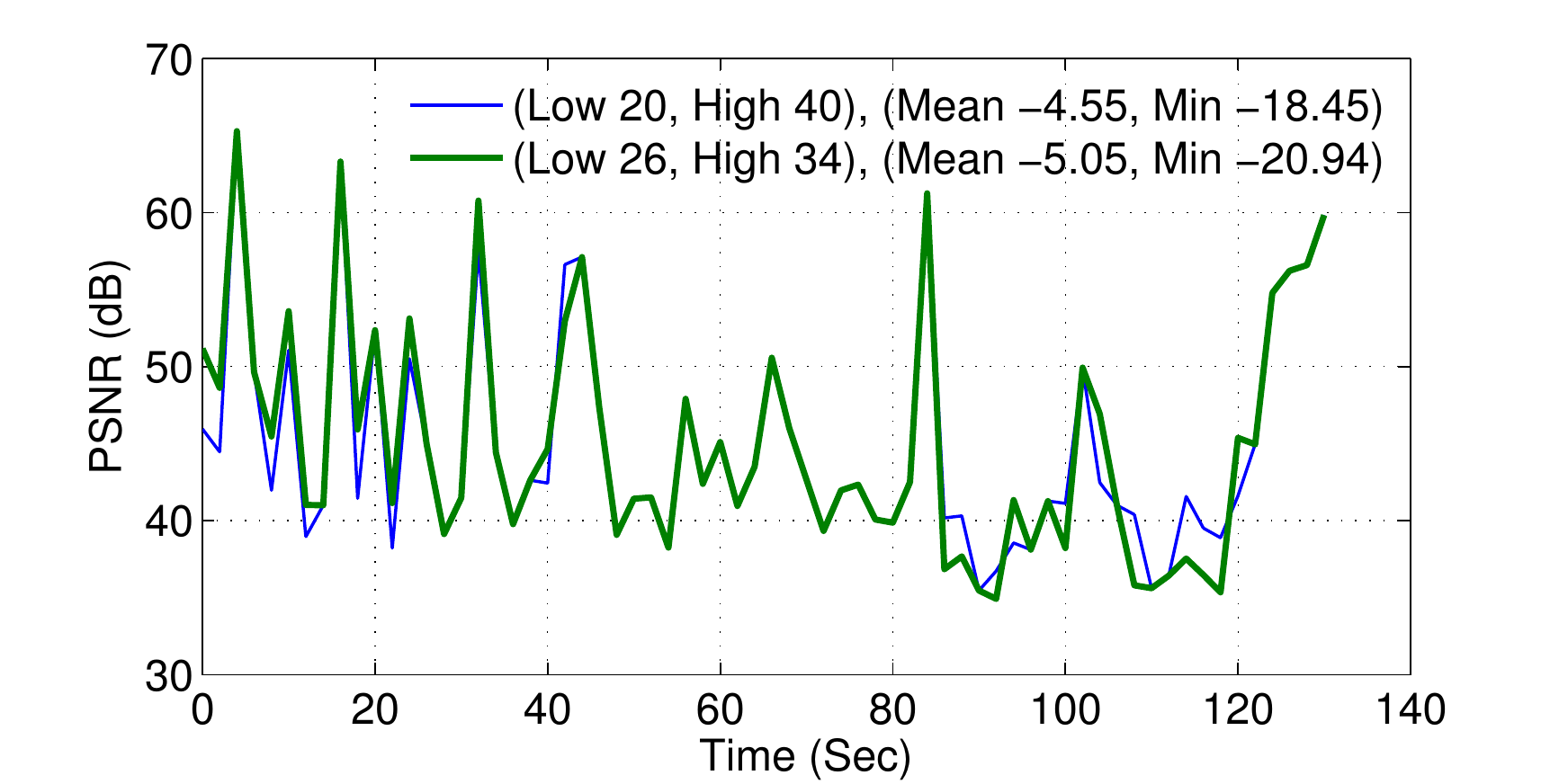} 
\par\end{center}

\begin{center}
\vspace{-0.05in}
 {\footnotesize (a1) Quality (}\emph{\footnotesize Elysium)}{\footnotesize{} }
\par\end{center}{\footnotesize \par}

\begin{center}
{\footnotesize \vspace{-0.25in}
}
\par\end{center}{\footnotesize \par}

\begin{center}
{\footnotesize \includegraphics[scale=0.35]{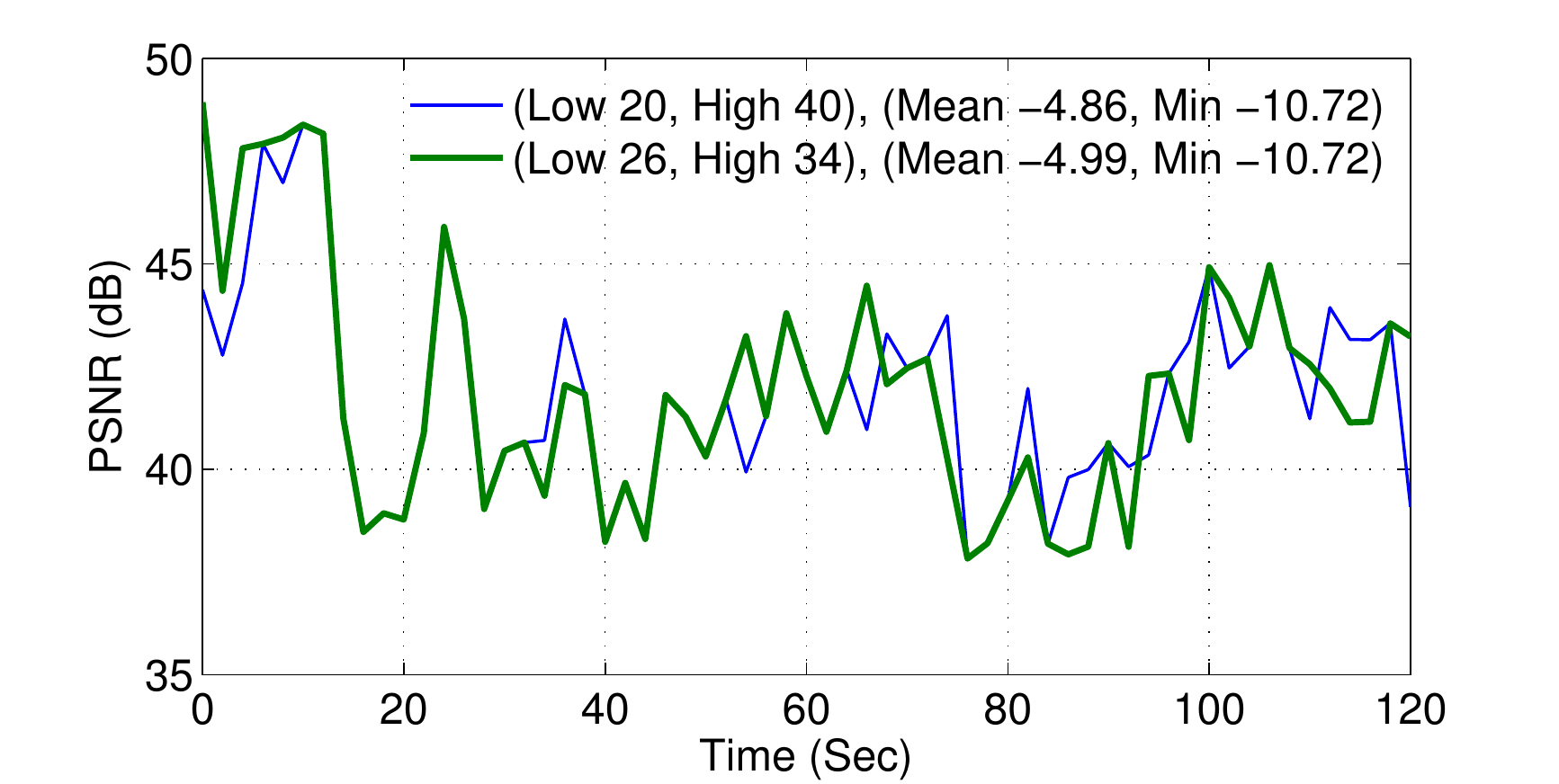} }
\par\end{center}{\footnotesize \par}

\begin{center}
{\footnotesize \vspace{-0.05in}
 (b1) Quality (}\emph{\footnotesize Avatar}{\footnotesize ) }
\par\end{center}%
\end{minipage}%
\begin{minipage}[t]{0.8\columnwidth}%
\begin{center}
\includegraphics[scale=0.35]{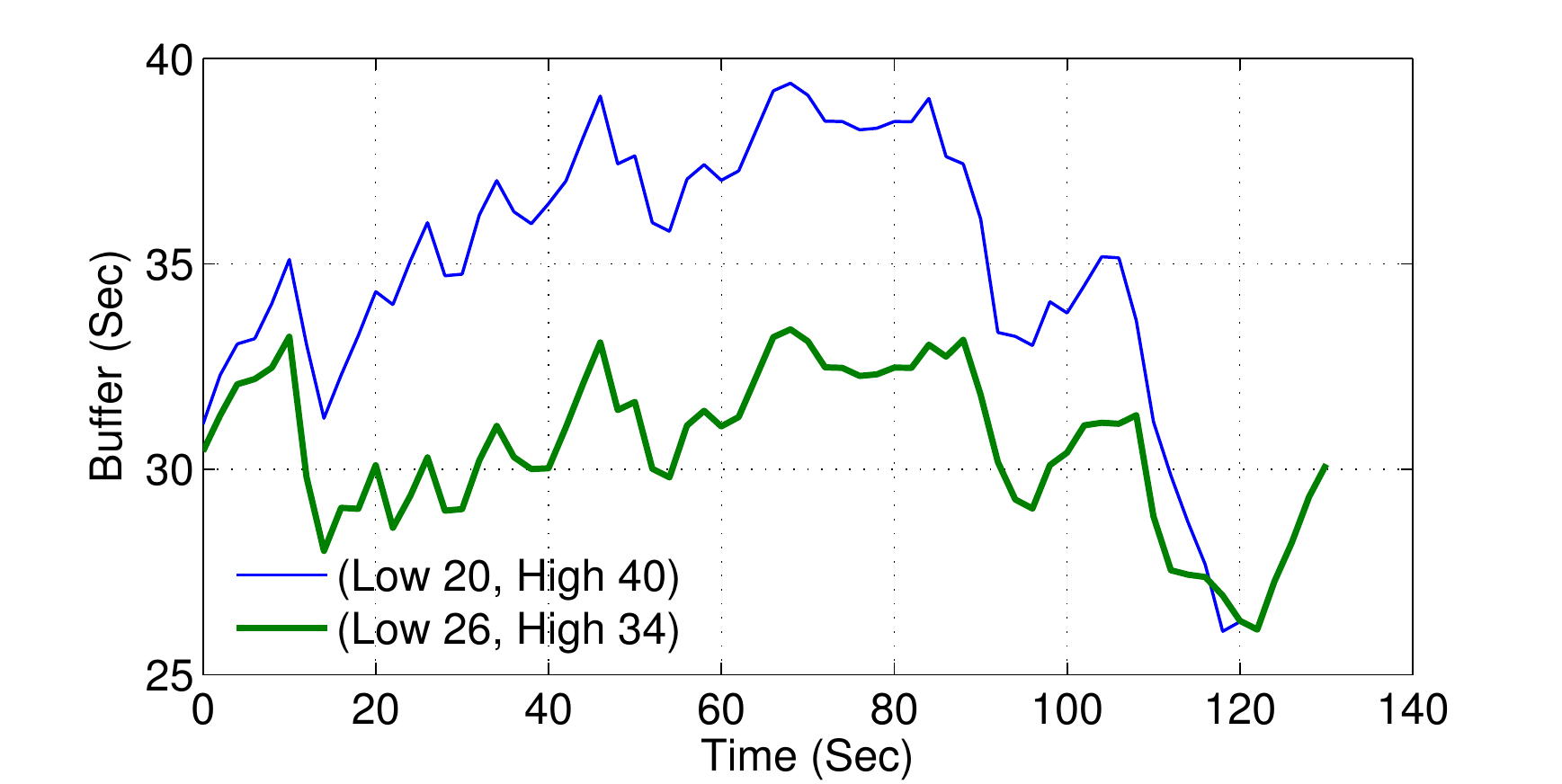} 
\par\end{center}

\begin{center}
\vspace{-0.05in}
 {\footnotesize (a2) Buffer (}\emph{\footnotesize Elysium}{\footnotesize ) }
\par\end{center}{\footnotesize \par}

\begin{center}
{\footnotesize \vspace{-0.25in}
}
\par\end{center}{\footnotesize \par}

\begin{center}
{\footnotesize \includegraphics[scale=0.35]{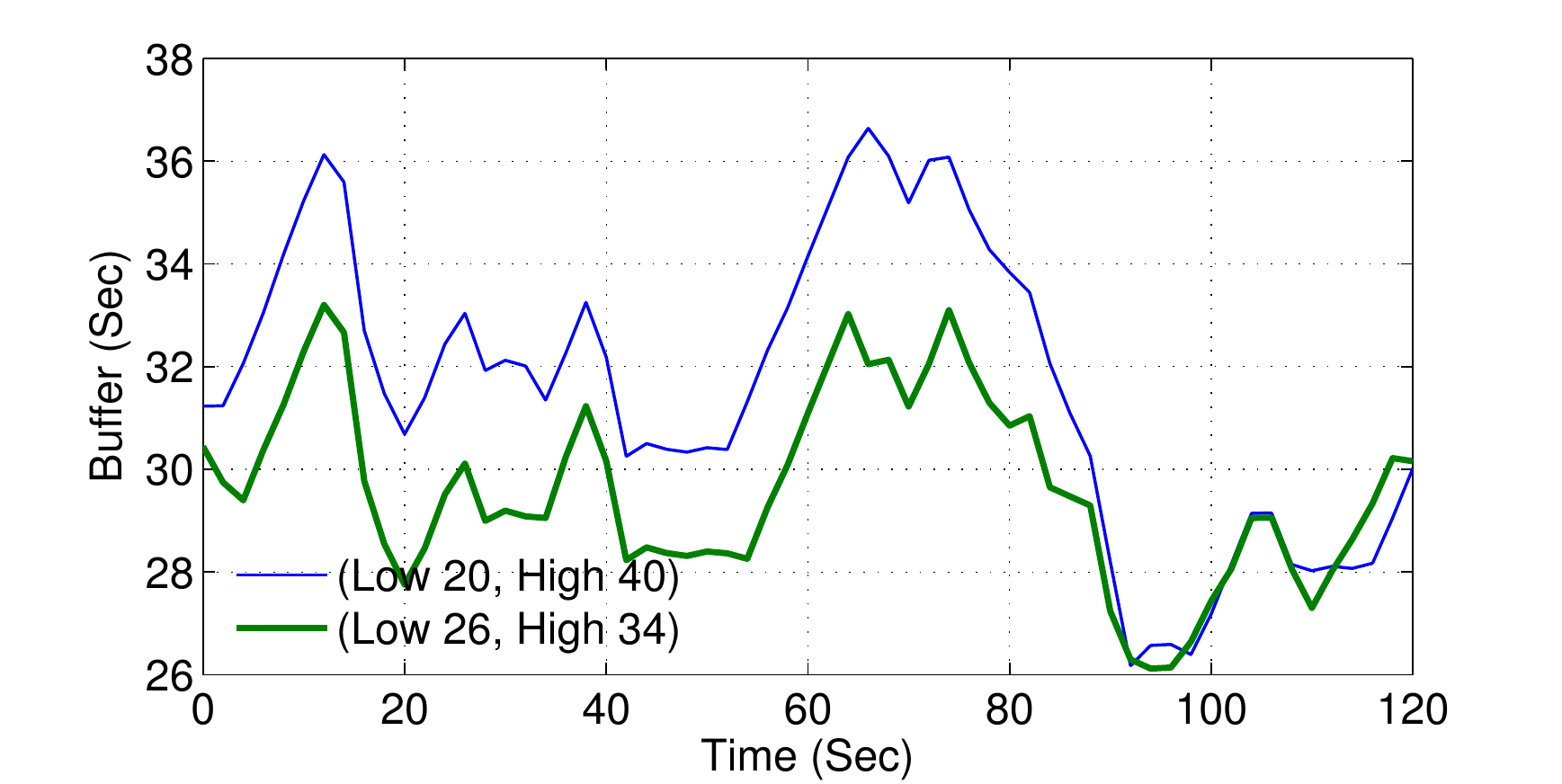} }
\par\end{center}{\footnotesize \par}

\begin{center}
{\footnotesize \vspace{-0.05in}
 (b2) Buffer (}\emph{\footnotesize Avatar}{\footnotesize ) }
\par\end{center}%
\end{minipage}
\par\end{centering}

\vspace{0.05in}

\caption{Comparing the traces of dynamic programming solution that maximizes
the mean quality with buffer lower and upper bound of 1) 20 and 40
seconds, respectively, and 2) 34 and 40 seconds, respectively. The
available bandwidth is set at constant 0.81 Mbps for \emph{Elysium}
and 2.60 Mbps for \emph{Avatar}. The initial and final buffer levels
are 30 seconds. The reported quality in -MSE is converted to PSNR
using (\ref{eq:psnr}) for better display.}

\label{Flo:exp_buffer} \vspace{-0.1in}
 
\end{figure*}

\section{Performance Evaluation}

\label{sec:Performance-Evaluation}

In this section, we evaluate the performance of the proposed algorithms.
The goal is to understand the behavior of individual modules, as well
as their aggregate performance. In the first step, we evaluate the
dynamic programming solution and the online algorithm in MATLAB simulations.
In the second step, we integrate them into the PANDA algorithm and
evaluate the performance in the ns-2 simulator \cite{ns2}. Besides
the simulation results shown in this section, we also provide a few
sample videos online \cite{cqsamples} for readers' subjective evaluation.

We have identified several existing quality-based HAS schemes \cite{mehrotra09,Jarnikov:SPIC11,vinay2013,BT12,georgopoulos2013towards}
(refer to Section \ref{sec:Related-Work} for discussions). However,
they either focus on a different perspective (e.g., encoding, cross-stream
optimization), or are based on different assumptions (e.g., scalable
coded video source, statistically stationary source/channel models).
Thus, it is not possible to directly compare our scheme with them.
Instead, we compare our solution with the bitrate-based adaptation
scheme that is unaware of the video quality information.

\subsection{Simulation Setup}

We select two video sources for our evaluation. The first one is a
two-minute long 720p \emph{Elysium} trailer crawled from YouTube \cite{cqsamples}.
The second one is a twelve-minute long 1080p clip extracted from the
movie \emph{Avatar}. For MATLAB evaluation, we use a two-minute part
of the \emph{Avatar} clip. Each video is chopped into segments of
two seconds. The \emph{Elysium} clip is encoded into seven bitrate
levels 400, 600, 800, 1200, 1600, 2400 and 3200 Kbps, and the \emph{Avatar}
clip in 11 bitrate levels 400, 600, 800, 1200, 1600, 2400, 3200, 4400,
5600, 7000 and 9000 Kbps. 

At each level, we use CBR encoding. This is to illustrate that our
client algorithm does not require the video to be VBR-encoded, and
also for a fair comparison with the bitrate-based adaptation scheme.
Note that the proposed algorithm is orthogonal to how the videos are
encoded. In practice, we find that (capped) VBR-encoded video content
generally works better with our client algorithm than CBR.

To measure the video quality, we simply use the \emph{negation} of
mean-squared error (MSE) value for each segment. Note that while it
may not be the metric to produce the best visual quality, it is good
enough for numerically comparing different schemes. However, in the
plots, the MSE values are converted to PSNR for better display, using
formula \cite{psnr}:
\begin{equation}
PSNR=10\cdot\log_{10}\left(\frac{255^{2}}{MSE}\right).\label{eq:psnr}
\end{equation}

\begin{figure*}
\begin{centering}
\begin{minipage}[t]{0.8\columnwidth}%
\begin{center}
\includegraphics[scale=0.35]{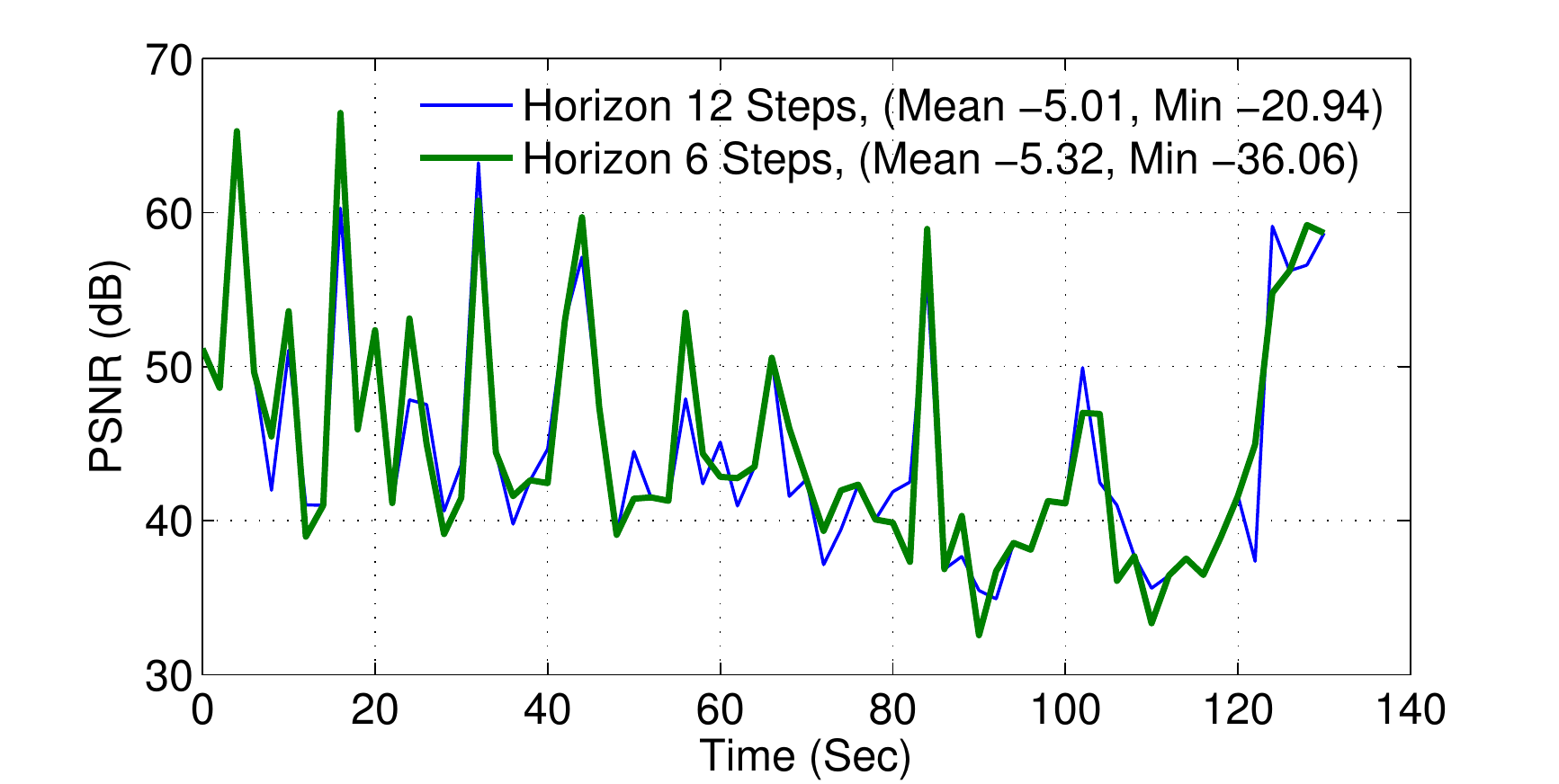} 
\par\end{center}

\begin{center}
\vspace{-0.05in}
 {\footnotesize (a1) Quality (}\emph{\footnotesize Elysium)}{\footnotesize{} }
\par\end{center}{\footnotesize \par}

\begin{center}
{\footnotesize \vspace{-0.25in}
}
\par\end{center}{\footnotesize \par}

\begin{center}
{\footnotesize \includegraphics[scale=0.35]{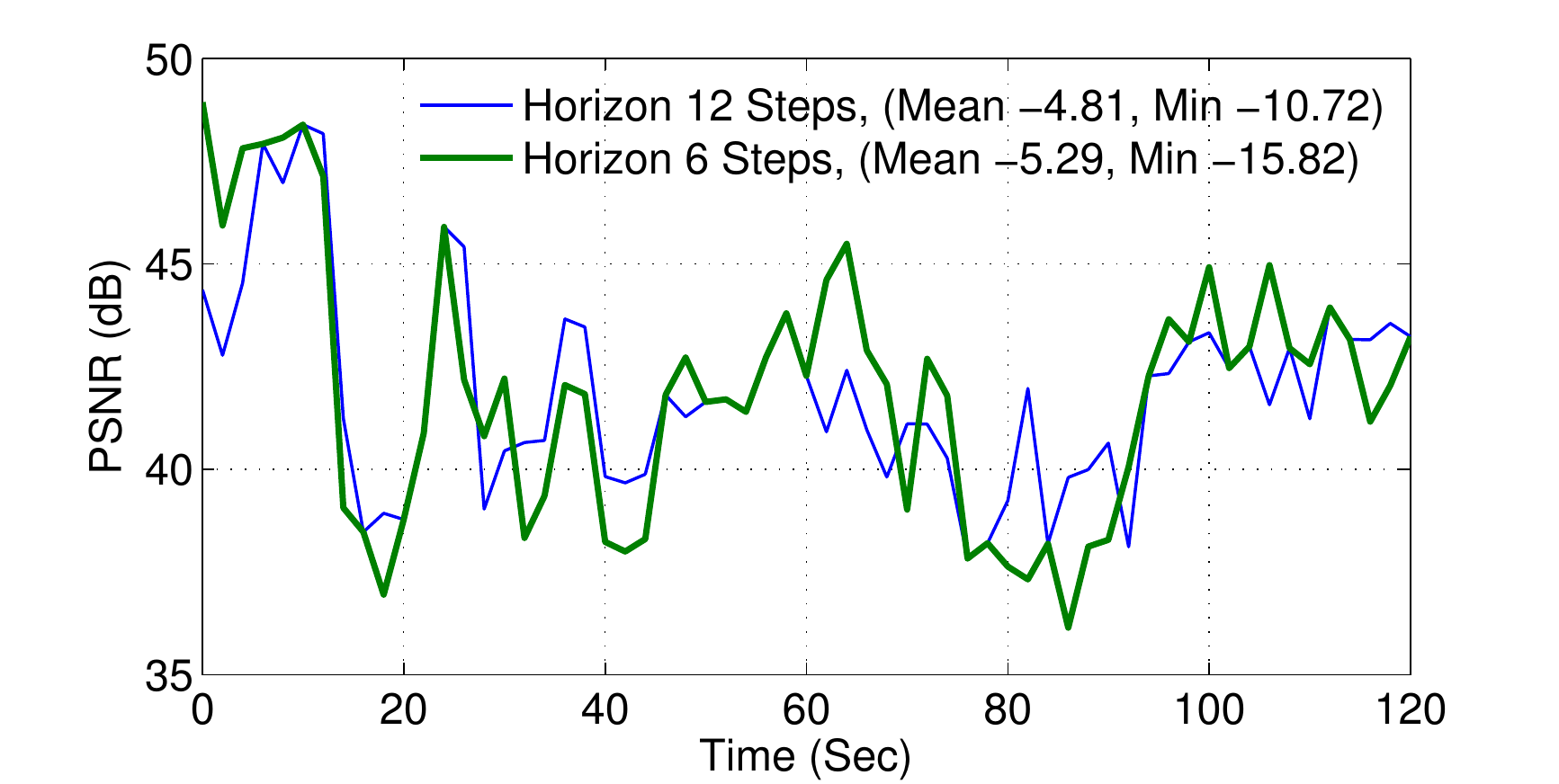} }
\par\end{center}{\footnotesize \par}

\begin{center}
{\footnotesize \vspace{-0.05in}
 (b1) Quality (}\emph{\footnotesize Avatar}{\footnotesize ) }
\par\end{center}%
\end{minipage}%
\begin{minipage}[t]{0.8\columnwidth}%
\begin{center}
\includegraphics[scale=0.35]{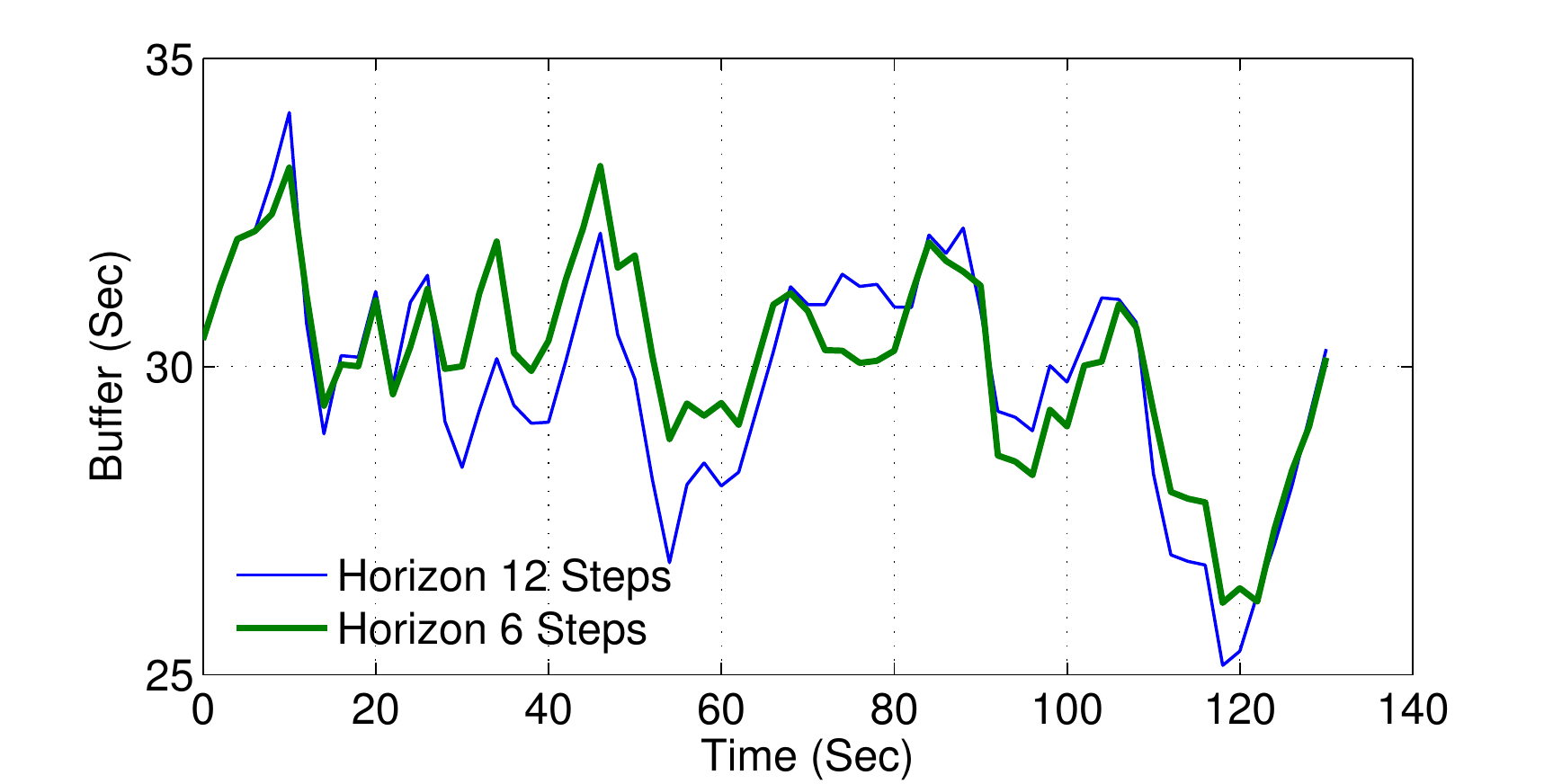} 
\par\end{center}

\begin{center}
\vspace{-0.05in}
 {\footnotesize (a2) Buffer (}\emph{\footnotesize Elysium}{\footnotesize ) }
\par\end{center}{\footnotesize \par}

\begin{center}
{\footnotesize \vspace{-0.25in}
}
\par\end{center}{\footnotesize \par}

\begin{center}
{\footnotesize \includegraphics[scale=0.35]{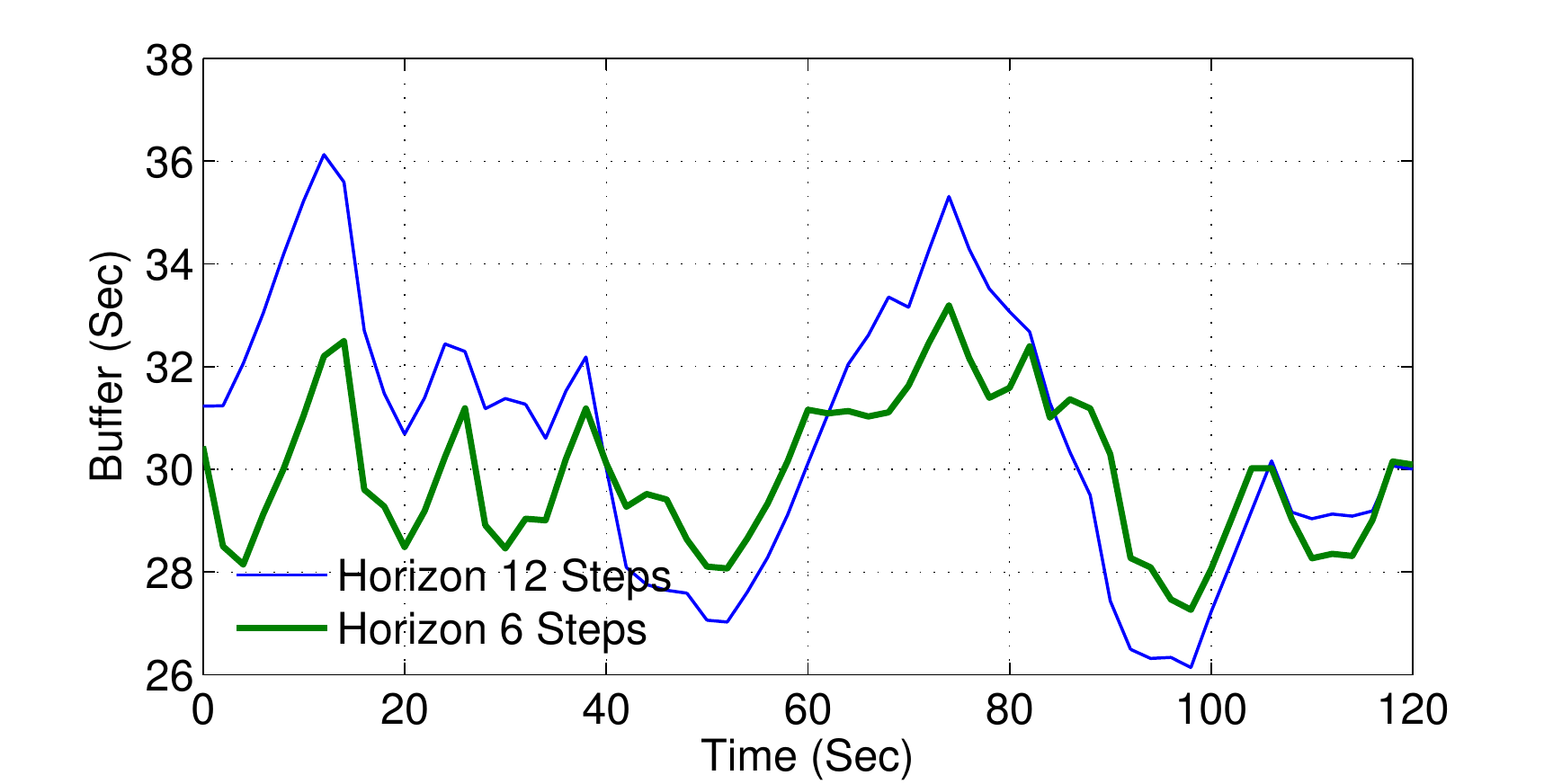} }
\par\end{center}{\footnotesize \par}

\begin{center}
{\footnotesize \vspace{-0.05in}
 (b2) Buffer (}\emph{\footnotesize Avatar}{\footnotesize ) }
\par\end{center}%
\end{minipage}
\par\end{centering}

\vspace{0.05in}

\caption{Comparing the traces of online algorithm that maximizes the mean quality
with finite horizon of size 1) 12 steps and 2) 6 steps. The available
bandwidth is set at constant 0.81 Mbps for \emph{Elysium} and 2.60
Mbps for \emph{Avatar}. The initial and final buffer levels are 30
seconds. Buffer lower and upper bounds are 20 and 40 seconds, respectively.
The reported quality in -MSE is converted to PSNR using (\ref{eq:psnr})
for better display.}

\label{Flo:exp_horizon2} \vspace{-0.1in}
 
\end{figure*}

In the MATLAB simulation, we input the bitrate and quality information
of the segments into the program. Assuming that we know the network
bandwidth and there are no gaps between segment downloads, we can
precisely calculate the evolution of the client buffer and perform
optimization accordingly.

In the ns-2 simulation, we evaluate the scenario of multiple clients
sharing a bottleneck link. The network is configured as in Figure
\ref{Flo:testbed_network}. The queueing policy used at the aggregation
router-home router bottleneck link is the following. For a link bandwidth
below or equal to 20 Mbps, we use random early detection (RED) with
$(min\_thr,max\_thr,p)=(30,90,0.25)$. The default parameters used
in the PANDA/CQ and the original PANDA algorithms are listed in Table
\ref{tab:parameters}.

\subsection{Dynamic Programming}

In the first experiment to evaluate the dynamic programming solution,
we compare three schemes: 1) bitrate-based fetching that is unaware
of the quality information, 2) dynamic programming that maximizes
the minimal quality and 3) dynamic programming that maximizes the
mean quality (equivalently, the total quality). We set the lower and
upper bounds of the buffer to be loose (20 and 50 seconds, respectively)
so that we can see the best gain achievable by quality optimization.

Figure \ref{Flo:exp_optobj} shows the traces of the quality (converted
to PSNR for better display) and the bitrate of fetched segments for
the three schemes. From the quality trace, we can observe that the
two quality-optimized schemes yield much better quality than the quality-unaware
scheme, both in terms of mean quality and minimal quality. The scheme
maximizing the minimal quality achieves best minimum quality (e.g.,
$-10.02$ for \emph{Elysium} compared to $-58.49$ of the quality-unaware
scheme), and the scheme maximizing the mean quality achieves best
mean quality (e.g., $-4.46$ for \emph{Elysium} compared to $-8.71$
of the quality-unaware scheme).

We are interested in how the buffer constraint would impact the quality
optimization. In the next experiment, we keep the objective to be
maximizing the mean quality, and vary the buffer bound in the dynamic
programming solution. We test two sets of lower an upper bounds: $(B_{L},B_{H})=(20,40)$
and $(B_{L},B_{H})=(26,34)$ seconds. The reference buffer level $B_{0}$
is set to 30 seconds. The resulting traces of quality, bitrate and
buffer evolution are shown in Figure \ref{Flo:exp_buffer}. From the
buffer evolution plot, we can verify that the resulting buffers are
strictly within the specified lower and upper bounds. From the quality
trace plot, we can see that the optimal mean quality decreases as
the upper bound becomes tighter, which well agrees with our intuition.
In Figure \ref{Flo:trend_buffer}, we show the trend of how the buffer
low and upper bound would affect the mean and minimum quality of the
two video sources. As the bound becomes loose, the quality improvement
will reach a saturation point beyond which further loosening the bound
would no longer improve the quality.

\begin{figure}
\begin{centering}
\begin{minipage}[t]{0.49\columnwidth}%
\begin{center}
\includegraphics[scale=0.36]{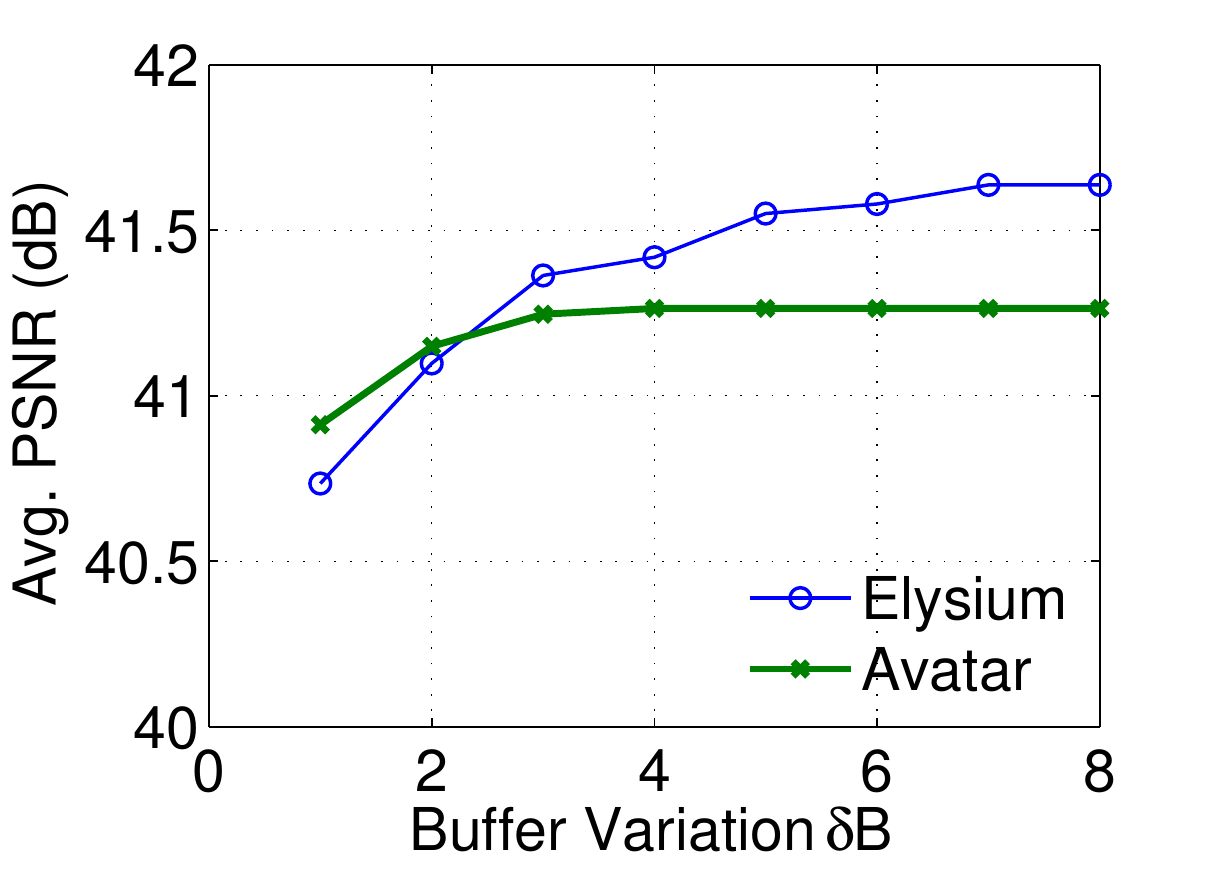} 
\par\end{center}

\begin{center}
\vspace{-0.05in}
 {\footnotesize (a) Mean Quality }
\par\end{center}{\footnotesize \par}

\begin{center}
{\footnotesize \vspace{-0in}
}
\par\end{center}%
\end{minipage}%
\begin{minipage}[t]{0.49\columnwidth}%
\begin{center}
\includegraphics[scale=0.36]{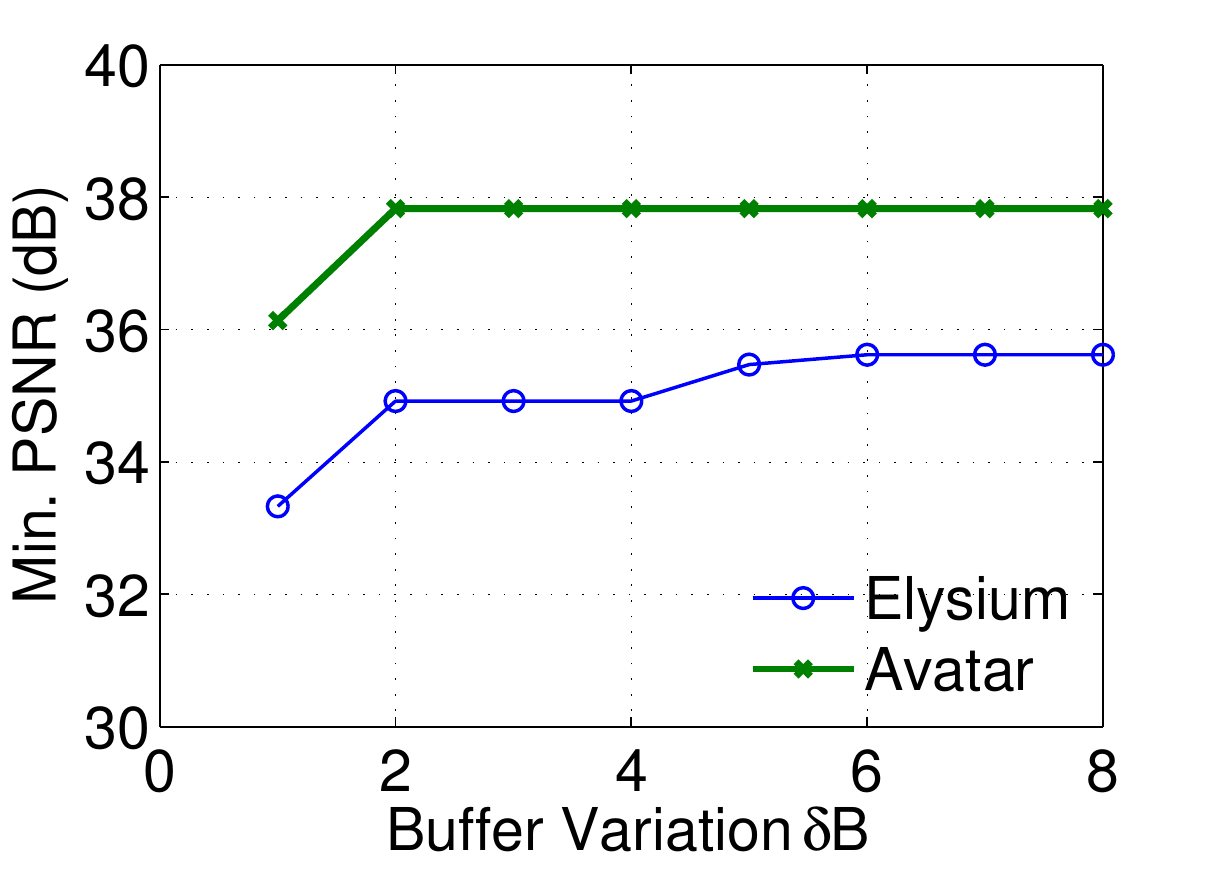} 
\par\end{center}

\begin{center}
\vspace{-0.05in}
 {\footnotesize (b) Minimum Quality }
\par\end{center}{\footnotesize \par}

\begin{center}
{\footnotesize \vspace{-0in}
}
\par\end{center}%
\end{minipage}
\par\end{centering}

\vspace{-0.15in}
 \caption{The mean and minimum quality as a function of the buffer upper and
lower bound $(B_{L},B_{H})=(30-2\cdot\delta B,30+2\cdot\delta B)$,
where $\delta B$ varies from 1 to 8. The initial and final buffer
levels are 30 seconds. }

\label{Flo:trend_buffer} \vspace{-0.1in}
 
\end{figure}

\subsection{Online Algorithm}

We proceed to evaluate the online algorithm, which uses the dynamic
programming solution as a building block. First, we would like to
evaluate how the size of the finite horizon would impact the quality
optimization result. we keep the objective to be maximizing the mean
quality, and compare the finite horizon size of 12 steps and 6 steps.
The resulting traces of quality and buffer evolution are shown in
Figure \ref{Flo:exp_horizon2}. From the quality trace plot, we can
see that the optimal mean quality decreases with a shorter horizon,
which well agrees with our intuition that myopic decision yields equal
or worse performance. From the buffer evolution plot, it is observed
that having a shorter horizon will limit the buffer's variability.
This is understood, as being myopic will limit the client to take
advantage of the buffer's breathing room. To see the general trend
of how the horizon size influences the mean and minimum quality of
the two video sources, refer to Figure \ref{Flo:trend_horizon}. Note
that the non-monotonic behavior may be due to the buffer quantization
effect, as discussed in Section \ref{sec:Dynamic-Programming-Algorithm}.
Similar to the buffer constraint, there is a similar saturation effect
in the horizon constraint, i.e., beyond certain point further improving
the horizon would no longer improve the video quality.

\begin{figure}
\begin{centering}
\begin{minipage}[t]{0.49\columnwidth}%
\begin{center}
\includegraphics[scale=0.36]{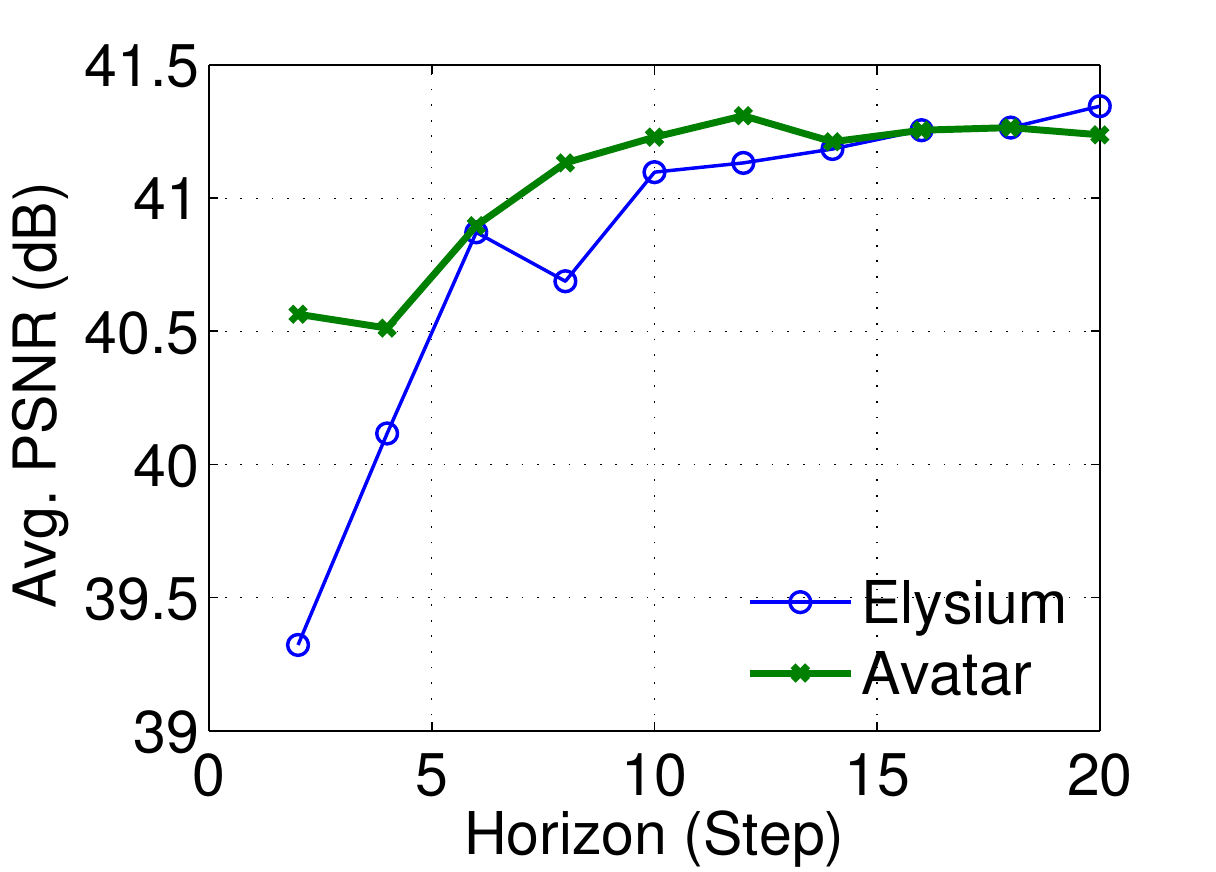} 
\par\end{center}

\begin{center}
\vspace{-0.05in}
 {\footnotesize (a) Mean Quality }
\par\end{center}{\footnotesize \par}

\begin{center}
{\footnotesize \vspace{-0in}
}
\par\end{center}%
\end{minipage}%
\begin{minipage}[t]{0.49\columnwidth}%
\begin{center}
\includegraphics[scale=0.36]{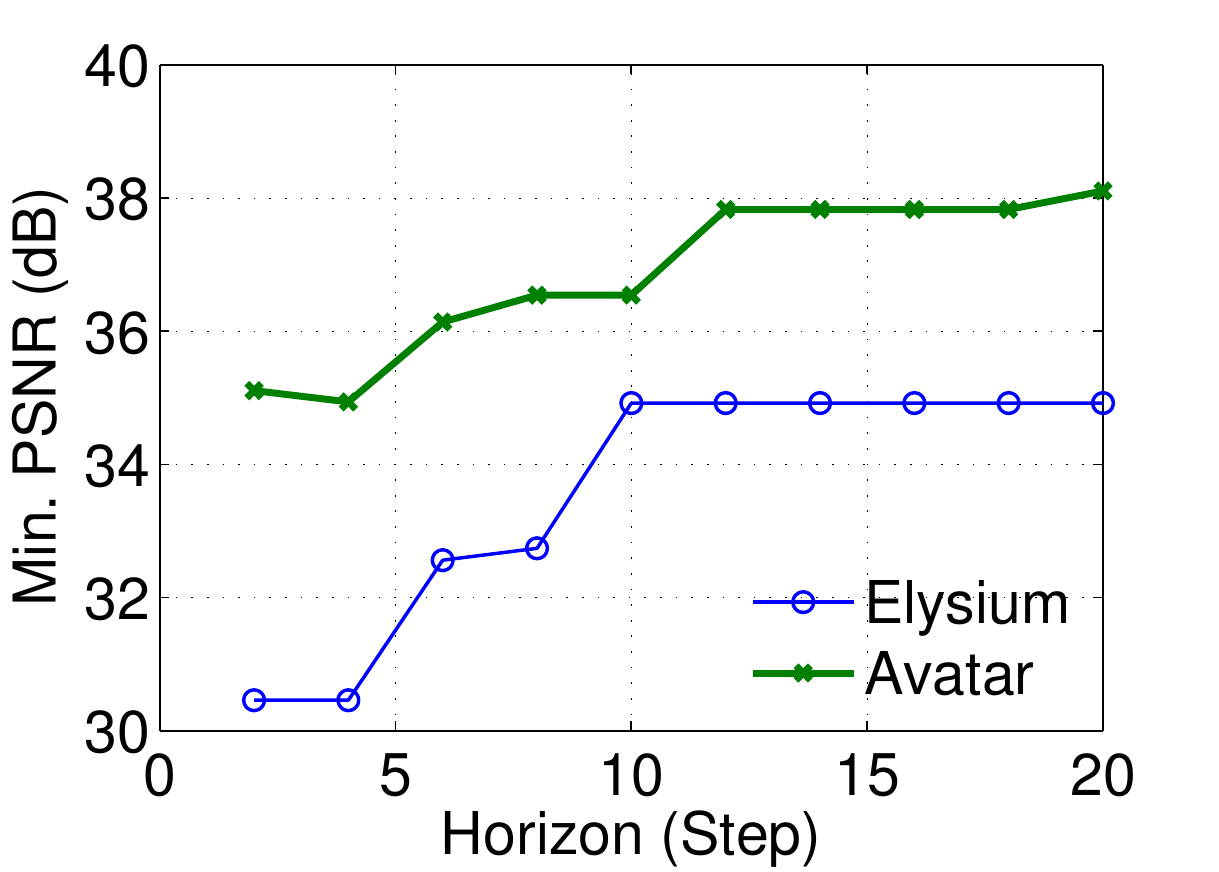} 
\par\end{center}

\begin{center}
\vspace{-0.05in}
 {\footnotesize (b) Minimum Quality }
\par\end{center}{\footnotesize \par}

\begin{center}
{\footnotesize \vspace{-0in}
}
\par\end{center}%
\end{minipage}
\par\end{centering}

\vspace{-0.15in}
 \caption{The mean and minimum quality as a function of the horizon size from
2 to 20 steps. The initial and final buffer levels are 30 seconds.
Buffer lower and upper bounds are 20 and 40 seconds, respectively.
The objective is to maximize the mean quality.}

\label{Flo:trend_horizon} \vspace{-0.1in}
 
\end{figure}

\begin{figure*}
\begin{centering}
\begin{minipage}[t]{0.66\columnwidth}%
\begin{center}
\includegraphics[scale=0.34]{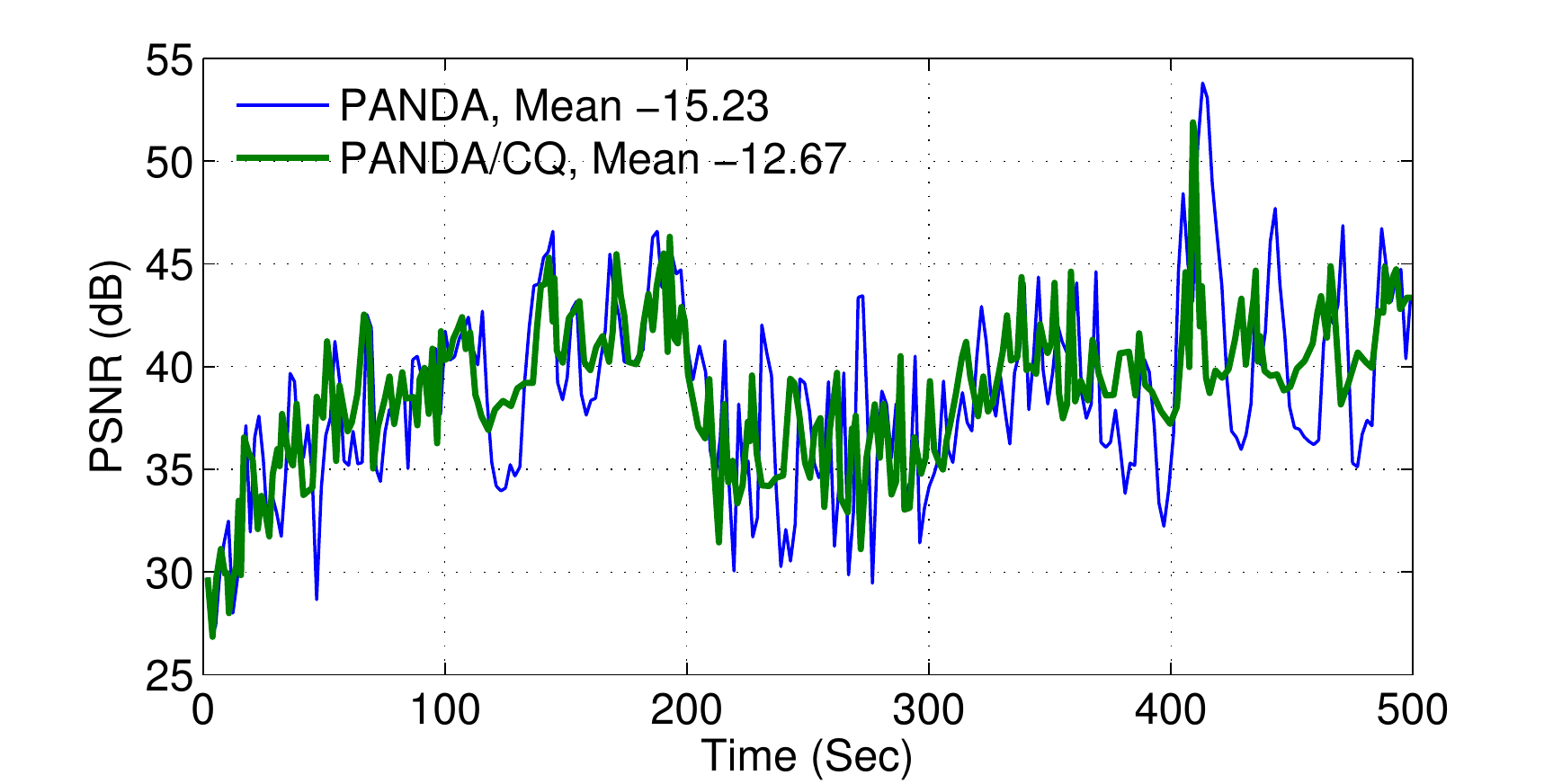} 
\par\end{center}

\begin{center}
\vspace{-0.1in}
 {\footnotesize (a) Quality }
\par\end{center}%
\end{minipage}%
\begin{minipage}[t]{0.66\columnwidth}%
\begin{center}
\includegraphics[scale=0.34]{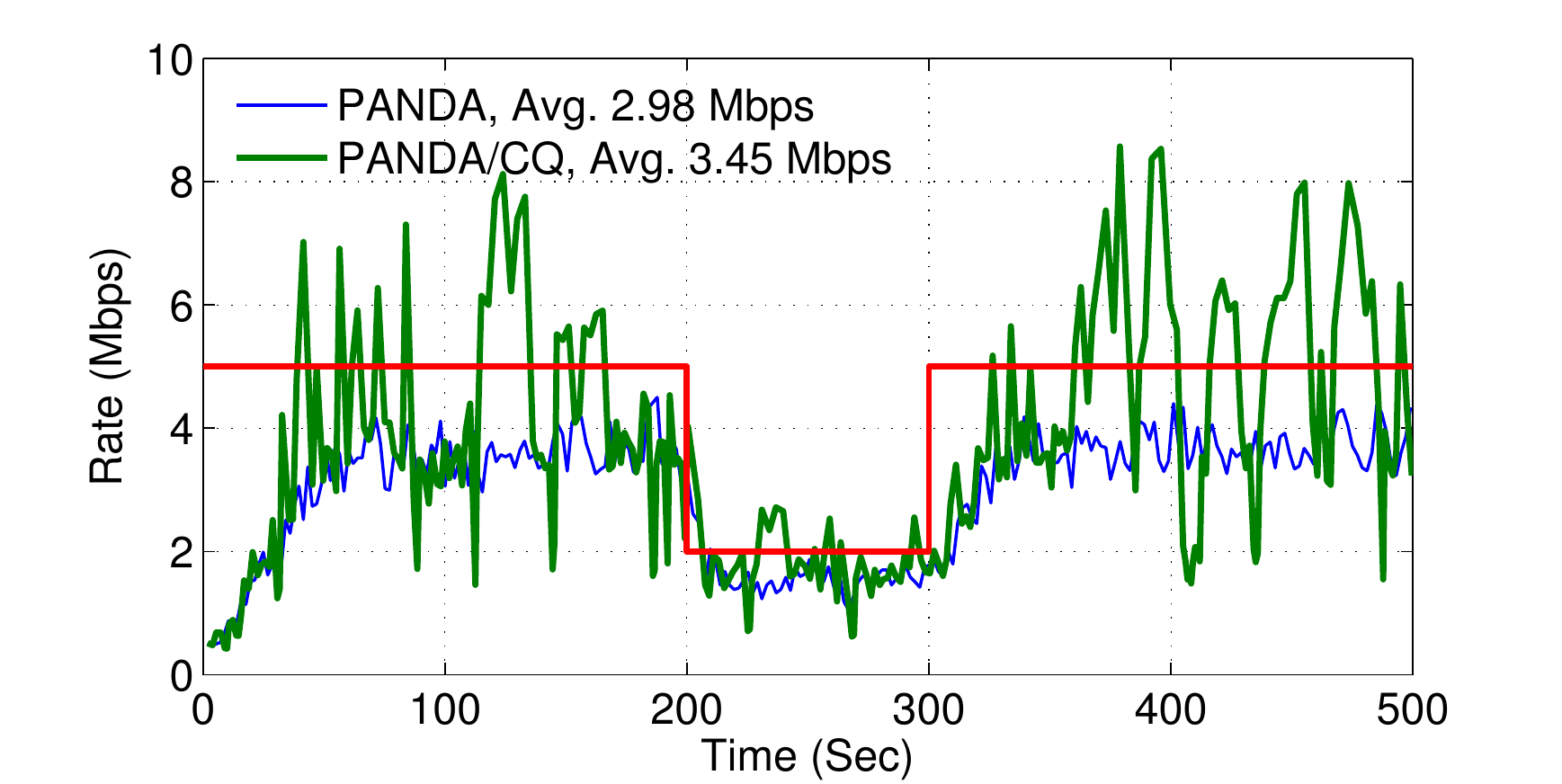} 
\par\end{center}

\begin{center}
\vspace{-0.1in}
 {\footnotesize (b) Rate }
\par\end{center}%
\end{minipage}%
\begin{minipage}[t]{0.66\columnwidth}%
\begin{center}
\includegraphics[scale=0.34]{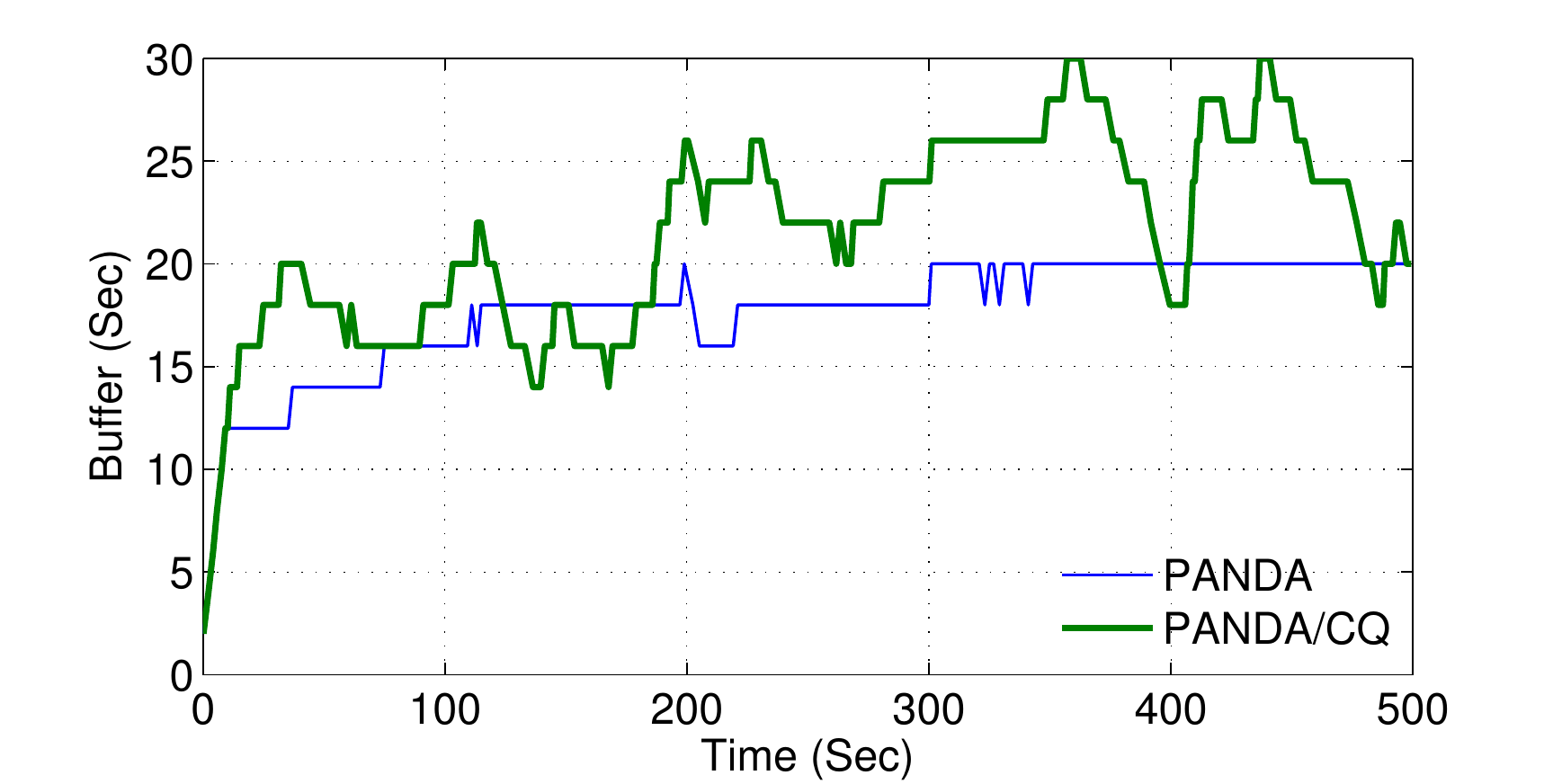} 
\par\end{center}

\begin{center}
\vspace{-0.1in}
 {\footnotesize (c) Buffer }
\par\end{center}%
\end{minipage}
\par\end{centering}

\vspace{-0.01in}

\caption{Comparing the traces of PANDA/CQ with the original PANDA algorithm
for a single client streaming with varying bandwidth. The bandwidth
is 5 Mbps for 0\textasciitilde{}200 seconds, 2 Mbps for 200\textasciitilde{}300
seconds and 5 Mbps for 300\textasciitilde{}500 seconds. Video source:
\emph{Avatar}. The reported quality in -MSE is converted to PSNR using
(\ref{eq:psnr}) for better display.}

\label{Flo:exp_pandacq_1} \vspace{-0.1in}
 
\end{figure*}

\subsection{PANDA/CQ}

Next, we integrate the dynamic programming solution and the online
algorithm into the PANDA rate adaptation algorithm, and examine the
aggregate behavior of the PANDA/CQ client. Throughout this subsection,
we use the objective of maximizing the mean quality. Two things that
we are most interested in are: 1) How does the algorithm respond to
bandwidth variation? 2) Can the PANDA/CQ client sustain similar stability
as the original PANDA?

We first examine the behavior of a single client under variable bandwidth.
We compare PANDA/CQ with PANDA under bandwidth variation from 5 Mbps
to 2 Mbps and to 5 Mbps (same setting as in \cite{panda13}). For
fairness, we set the multiplicative safety margin $\epsilon$ of PANDA
to be 0, and a lower reference buffer $B_{0}$ of 20 seconds. The
resulting traces are compared in Figure \ref{Flo:exp_pandacq_1}.
From the rate plot, we note that both algorithms are able to closely
track the bandwidth variation, thanks to the probing-and-adapt mechanism.
PANDA/CQ has a much larger variation of bitrate than PANDA, as its
rate adaptation takes into consideration the video content variability.
Accordingly, from the buffer plot, the buffer of PANDA/CQ fluctuates
within a bounded region; in contrast, the buffer of PANDA stays constant
at the reference level.

An important fact to notice from the rate plot is that PANDA/CQ has
a higher average fetching bitrate than PANDA (even with $\epsilon=0$).
The reason behind is that when PANDA/CQ plans on which segment to
fetch, it takes into consideration multiple segments in the future.
The resulting multiplexing effect creates a more continuous decision
space for the PANDA/CQ client to reduce the off-intervals as much
as it can. In contrast, in PANDA, the coarse quantization of video
bitrate leads to a very discrete decision space, resulting in large
off-intervals and low bandwidth utilization. Consequently, from the
quality plot, we can see that the mean quality gain for the PANDA/CQ
algorithm is higher than the gain noticed in the previous MATLAB simulations,
because it is not only contributed by the optimization algorithm,
but also the higher bandwidth utilization.

\begin{figure*}
\begin{centering}
\begin{minipage}[t]{0.66\columnwidth}%
\begin{center}
\includegraphics[scale=0.34]{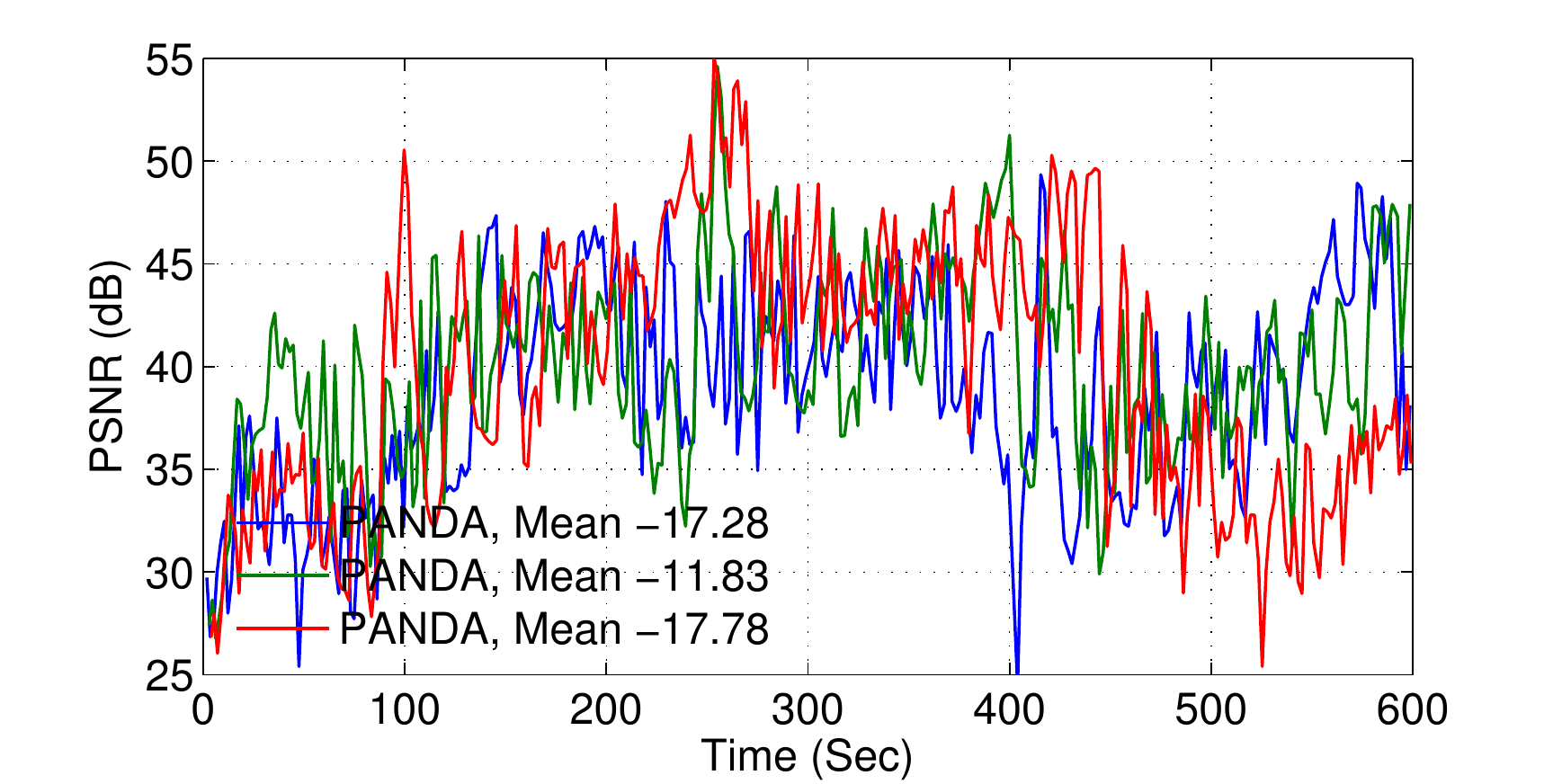} 
\par\end{center}

\begin{center}
\vspace{-0.1in}
 {\footnotesize (a1) Quality (PANDA) }
\par\end{center}%
\end{minipage}%
\begin{minipage}[t]{0.66\columnwidth}%
\begin{center}
\includegraphics[scale=0.34]{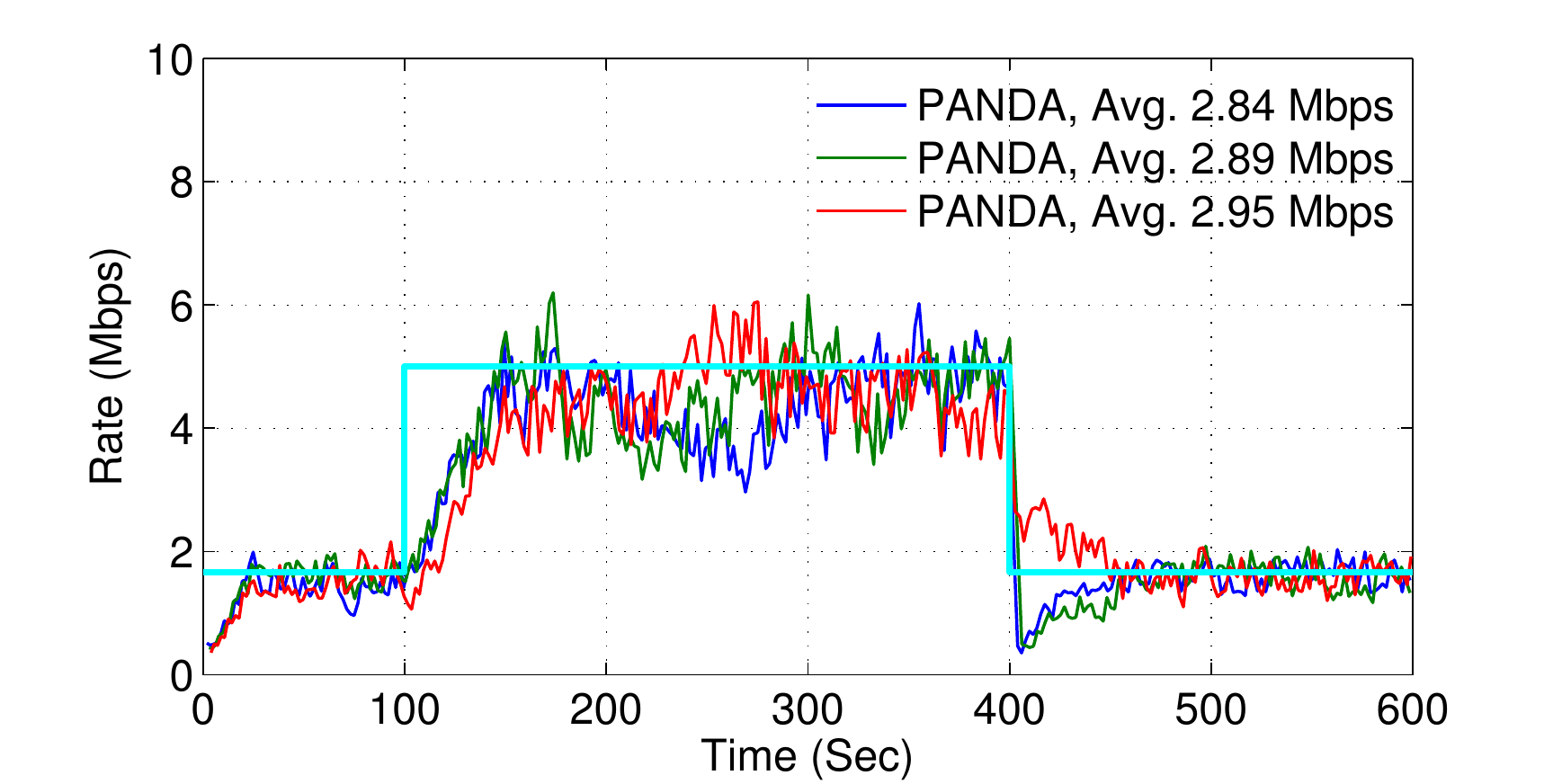} 
\par\end{center}

\begin{center}
\vspace{-0.1in}
 {\footnotesize (a2) Rate (PANDA) }
\par\end{center}%
\end{minipage}%
\begin{minipage}[t]{0.66\columnwidth}%
\begin{center}
\includegraphics[scale=0.34]{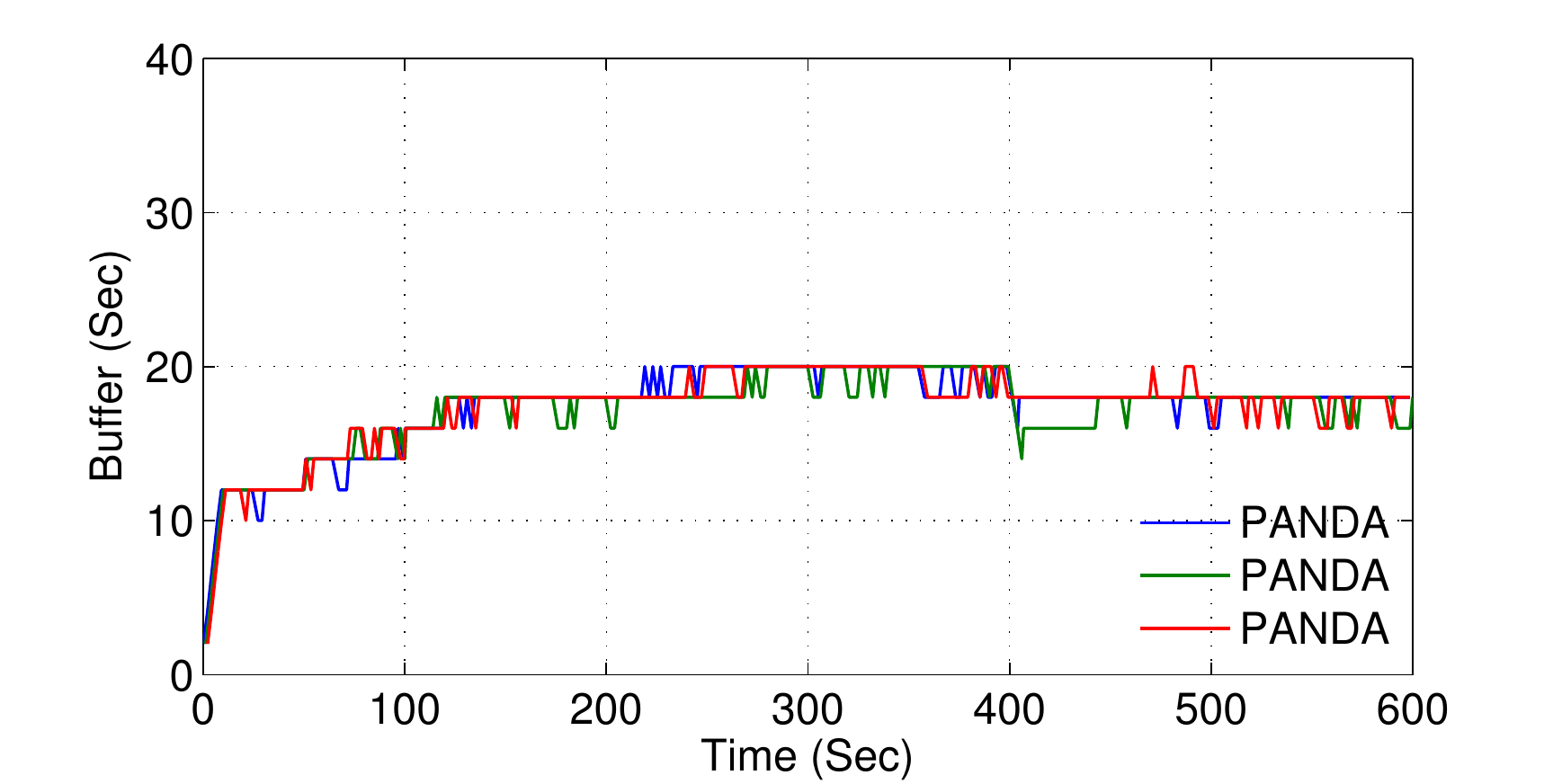} 
\par\end{center}

\begin{center}
\vspace{-0.1in}
 {\footnotesize (a3) Buffer (PANDA) }
\par\end{center}%
\end{minipage}
\par\end{centering}

\begin{centering}
\begin{minipage}[t]{0.66\columnwidth}%
\begin{center}
\includegraphics[scale=0.34]{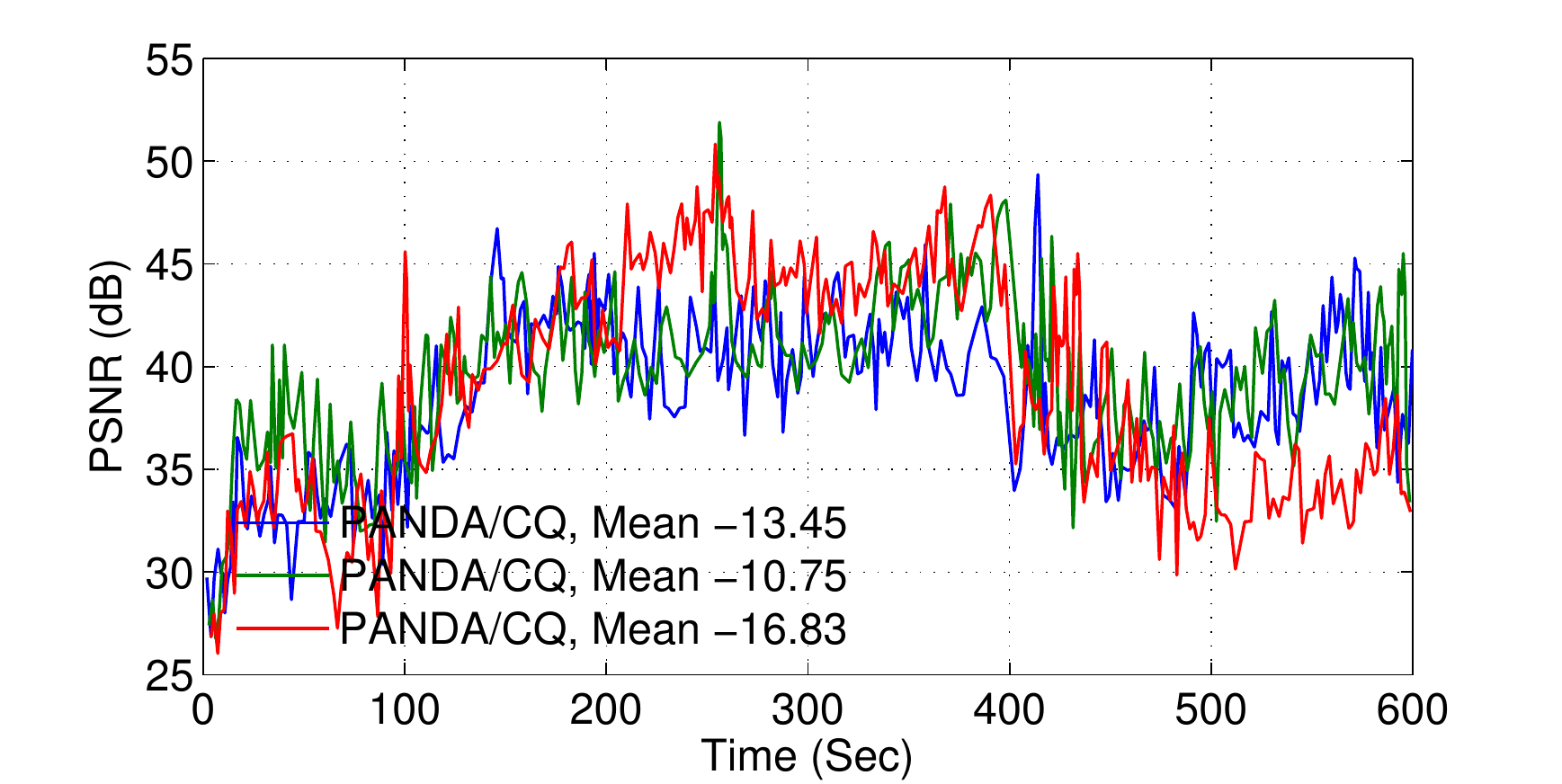} 
\par\end{center}

\begin{center}
\vspace{-0.1in}
 {\footnotesize (b1) Quality (PANDA/CQ) }
\par\end{center}%
\end{minipage}%
\begin{minipage}[t]{0.66\columnwidth}%
\begin{center}
\includegraphics[scale=0.34]{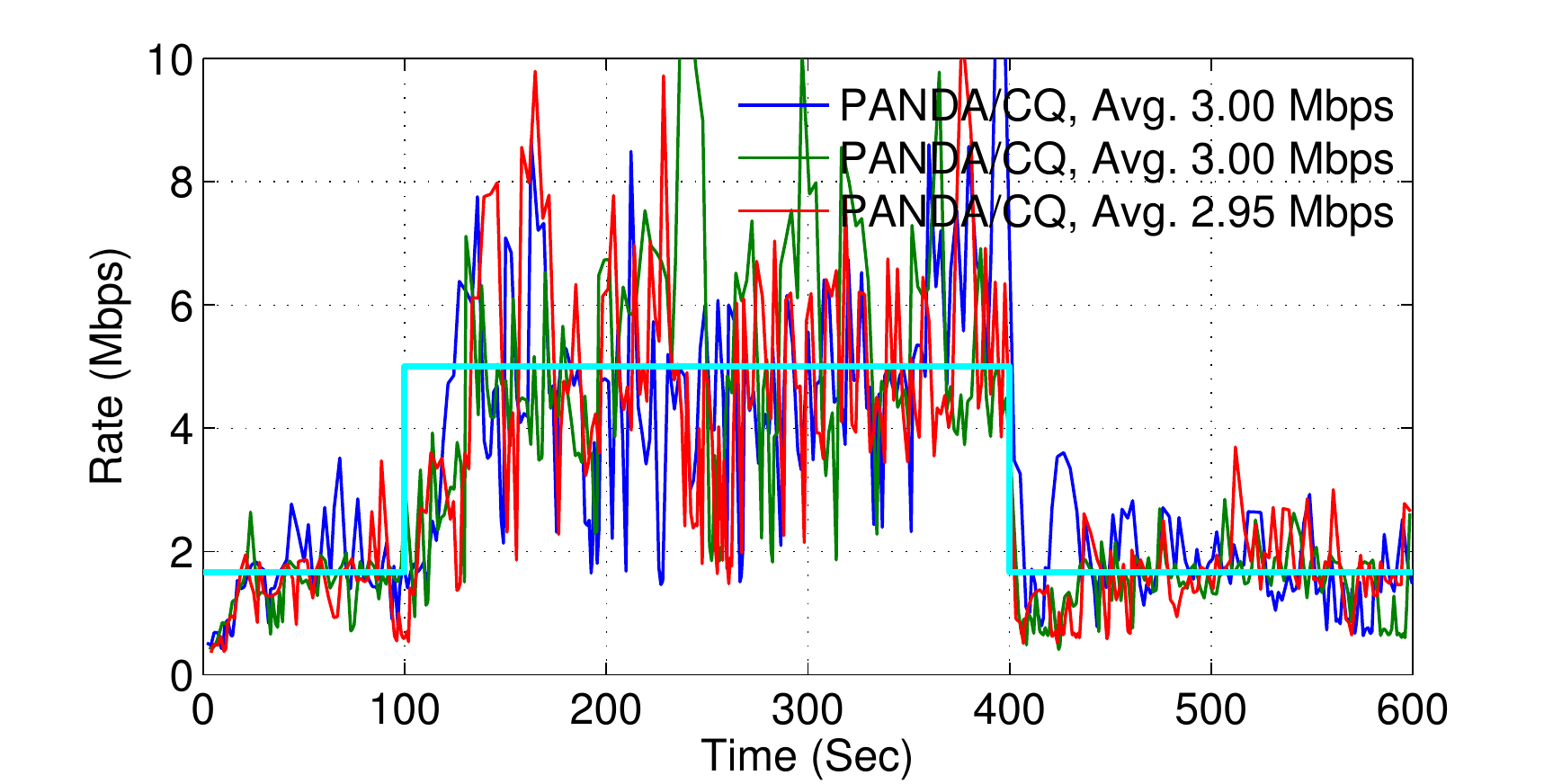} 
\par\end{center}

\begin{center}
\vspace{-0.1in}
 {\footnotesize (b2) Rate (PANDA/CQ) }
\par\end{center}%
\end{minipage}%
\begin{minipage}[t]{0.66\columnwidth}%
\begin{center}
\includegraphics[scale=0.34]{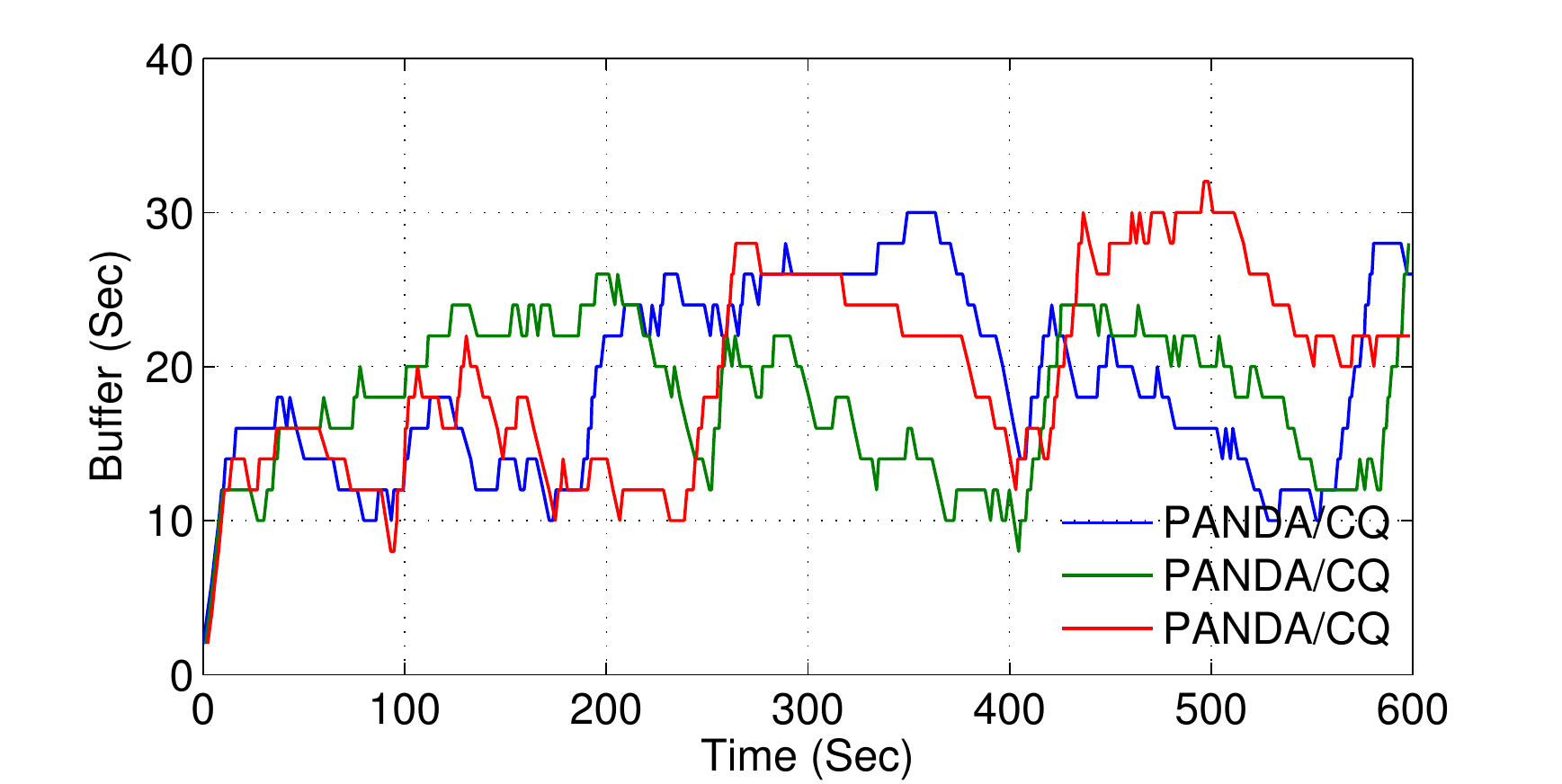} 
\par\end{center}

\begin{center}
\vspace{-0.1in}
 {\footnotesize (b3) Buffer (PANDA/CQ) }
\par\end{center}%
\end{minipage}
\par\end{centering}

\vspace{-0.01in}

\caption{Comparing the traces of PANDA/CQ with the original PANDA algorithm
for three clients sharing a link with varying bandwidth. The bandwidth
is 5 Mbps for 0\textasciitilde{}100 seconds, 15 Mbps for 100\textasciitilde{}400
seconds and 5 Mbps for 400\textasciitilde{}600 seconds. Video source:
\emph{Avatar}. Each client starts at a different position of the video
clip. The reported quality in -MSE is converted to PSNR using (\ref{eq:psnr})
for better display.}

\label{Flo:exp_pandacq_3} \vspace{-0.15in}
 
\end{figure*}

Next, we have multiple clients to compete at a bandwidth-varying link.
For each experiment, we have three PANDA or three PANDA/CQ clients
compete under bandwidth variation from 5 Mbps to 15 Mbps and to 5
Mbps. Each client start streaming the \emph{Avatar} video source from
a different position. The traces are shown in Figure \ref{Flo:exp_pandacq_3}.
From the quality plots, we can obviously observe that the PANDA/CQ
clients achieves higher mean quality and more stable quality than
the PANDA clients. The rate plots show that, similar to the single-client
case, PANDA/CQ clients are able to achieve higher bitrate and more
efficient link utilization. Lastly, from the buffer plots, we can
see that PANDA/CQ clients have their buffer fluctuate within a bounded
region whereas the PANDA clients have their buffers staying constant
at the reference level.

Lastly, we keep the link bandwidth to be constant and examine the
client behavior as we (a) vary the link bandwidth and (b) vary the
buffer lower bound $B_{L}$. We measure the 5-percentile PSNR of all
clients's downloaded segments, which considers both total quality
and quality variability. The trend plots are shown in Figure \ref{Flo:trend_pandacq_3}.
From (a), as we increase the link rate, the PANDA/CQ consistently
outperforms PANDA by more than 1 dB \emph{on average}. Note that the
\emph{worst-case} improvement is much greater (e.g., 5 dB). We find
that typically it is the worst-case improvement that dominates the
perceived visual quality. From (b), the 5-percentile PSNR decreases
as we tighten the buffer lower bound, but the minimum buffer increases.
Thus, we can see that $B_{L}$ is a parameter that controls the trade
off between video quality variability and the risk of buffer underrun.

\section{Related Work}

\label{sec:Related-Work}

\emph{Pre-HAS Video Streaming:} The literature on video streaming
techniques with quality optimization can be roughly categorized into
two eras -- the pre-HAS era and the post-HAS era. Early works (e.g.,
\cite{stuhmuller2000}) on video streaming assume generic lossy transmission
channel. For video streaming over packetized (e.g., IP) networks,
before the emergence of HAS, a common wisdom is to lay it on top of
lossy RTP/UDP to take advantage of the error-resilient nature of video
(e.g., \cite{radio06}) and apply error control as necessary. Thus,
a common theme in these works is to deal with quality degradation
caused by packet losses.

\emph{Post-HAS Video Streaming}: With the emergence of HAS, which
rides on top of TCP, packet loss is no longer a concern. Instead,
the main source of quality degradation becomes compression and downsampling
artifacts. There have been several on-going efforts trying to tackle
the video quality optimization problem for HAS, all from different
perspectives. Mehrotra and Zhao consider an approach based on rate-distortion
optimization and scalable video coding (SVC) \cite{mehrotra09}. They
formulate the problem with the buffer constraint in a way similar
to ours, and obtain a sub-optimal solution based on Lagrangian multiplier.
When attempting to extend their solution from SVC to redundantly encoded
multiple rate levels, they have noted that it yields incorrect answer
as the rate-distortion curve was not necessarily convex any more.
In contrast, our dynamic programming solution does not require convexity
in the rate-quality relationship.

In \cite{Jarnikov:SPIC11}, a Markov decision process (MDP) is used
to compute a set of optimal client strategies in order to maximize
the video quality. The MDP requires the knowledge of network conditions
and video content statistics, which may not be readily available.
Similar statistical and learning-based approaches are proposed by
Joseph and de Veciana \cite{vinay2013}. The optimality of their scheme
relies on strong statistical assumptions, such as stationary ergodicity
of the source and the channel. In contrast, as explained in the introduction
section, we have deliberately avoided a statistical model in this
work.

\begin{figure}
\begin{centering}
\begin{minipage}[t]{0.49\columnwidth}%
\begin{center}
\includegraphics[scale=0.36]{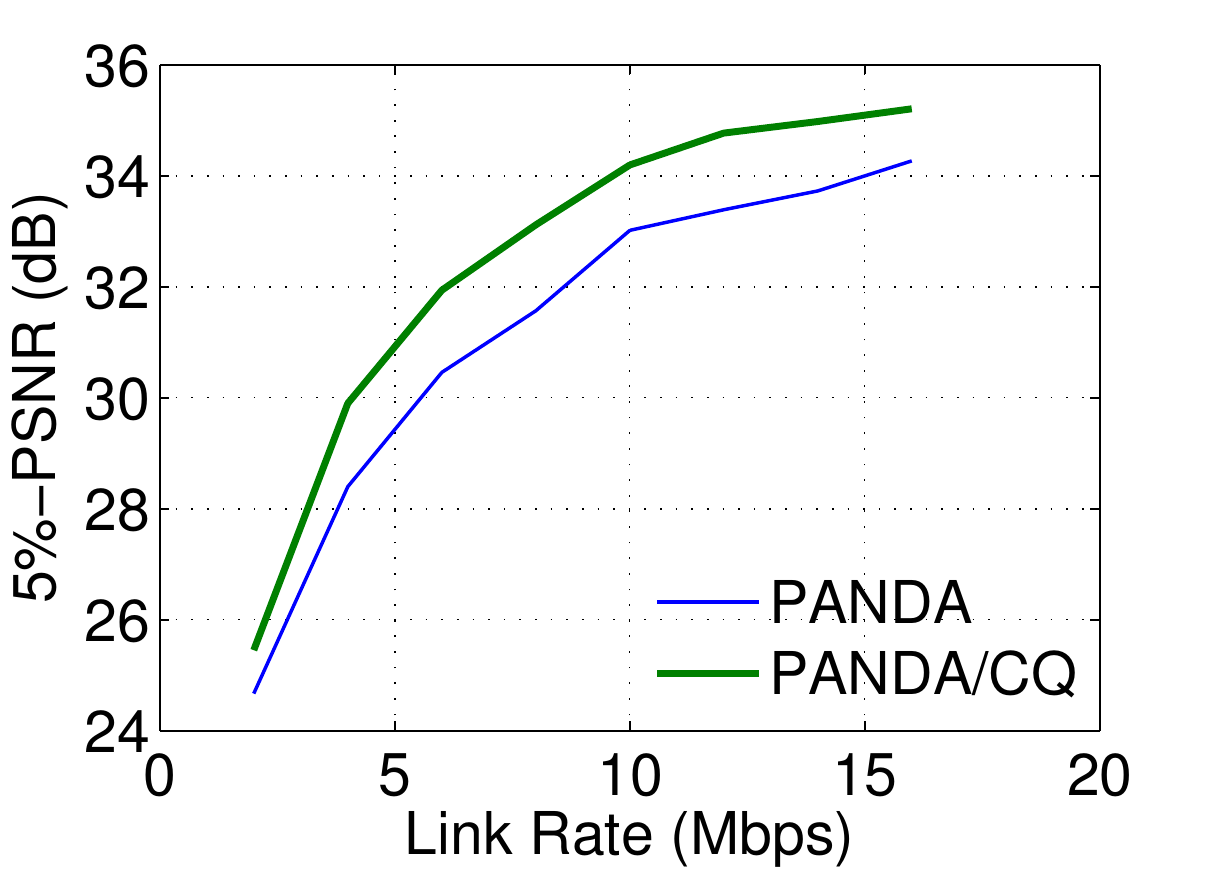} 
\par\end{center}

\begin{center}
\vspace{-0.05in}
 {\footnotesize (a) 5\%-PSNR vs. Link Rate }
\par\end{center}{\footnotesize \par}

\begin{center}
{\footnotesize \vspace{-0in}
}
\par\end{center}%
\end{minipage}%
\begin{minipage}[t]{0.49\columnwidth}%
\begin{center}
\includegraphics[scale=0.36]{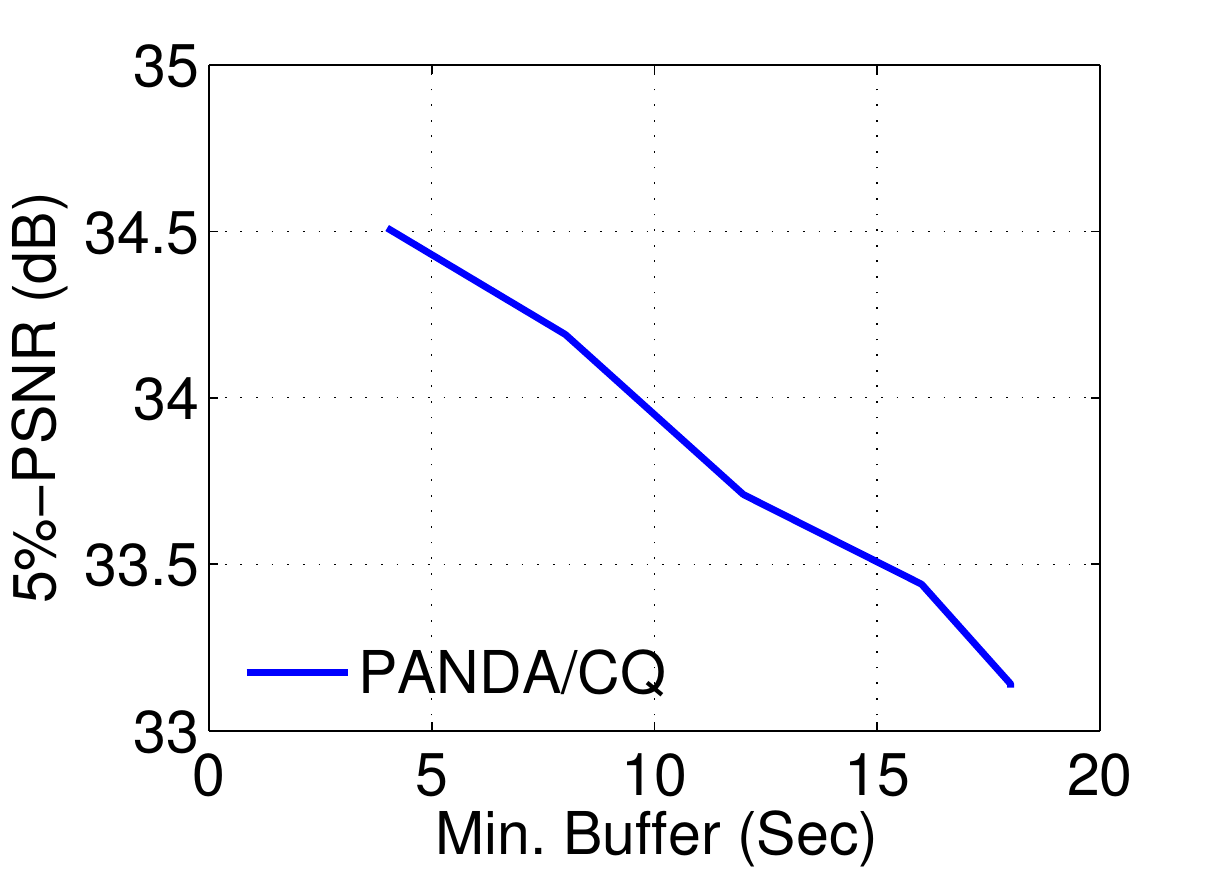} 
\par\end{center}

\begin{center}
\vspace{-0.05in}
 {\footnotesize (b) 5\%-PSNR vs. Minimum Buffer }
\par\end{center}{\footnotesize \par}

\begin{center}
{\footnotesize \vspace{-0in}
}
\par\end{center}%
\end{minipage}
\par\end{centering}

\vspace{-0.2in}
 \caption{Three clients share a constant-bitrate link. (a) 5-percentile PSNR
as a function of the link bitrate. (b) 5-percentile PSNR as a function
of the resulting minimum buffer, as we vary the buffer lower bound
from 4 seconds to 24 seconds.}

\label{Flo:trend_pandacq_3} \vspace{-0.1in}
 
\end{figure}

Crabtree et al. report the gains in terms of bitrate saved by using
a quality-optimized approach to HAS \cite{BT12}. Their technical
discussion mainly focuses on how to assemble a constant quality video
stream out of many CBR streams. Georgopoulos et al. study a network-based
approach to ensure the fairness of video quality among HAS streams
\cite{georgopoulos2013towards}. The multi-stream problem considered
is different from our work, as we focus on quality optimization within
a single stream.

There is also some ongoing standardization work in the MPEG. The DASH
working group is currently running a core experiment regarding quality-optimized
DASH streaming. The core experiment is still in progress, however,
it is expected to result in a signaling approach for carrying quality
and/or bitrate information at the segment level.

\emph{Video Quality Temporal Pooling}: On the study of temporal pooling
of video quality, a recent work \cite{bovik13} have shown that the
overall impression of a viewer towards a video is greatly influenced
by the single most severe event while the duration is neglected, which
corroborate our choice of the optimization objective. A more recent
study \cite{chen2013dynamic} dedicated to temporal pooling for HAS
proposes a more complicated linear dynamic system model with the intent
to capture the hysteresis effect in human visual response. Joseph
and de Veciana \cite{vinay2013} uses the difference between mean
quality and quality variability as the pooling metric.

\emph{Dynamic Programming: }Dynamic programming is a combinatorial
optimization technique that finds a wide range of engineering applications.
The application scenarios we have found that are most related to this
work are video encoding for CD-ROMs \cite{ortega1994optimal} and
quality control for scalable media processing \cite{wust2004quality}.

\section{Conclusion}

In this paper, we have proposed an optimization solution for streaming
video over HTTP with consistent quality. We have thoroughly examined
the designed algorithms, and integrated it into PANDA -- a practical
HAS rate adaptation algorithm for HAS deployment at large scale. The
proposed solution has the following features: 
\begin{itemize}
\item It operates independent of whether the source video is CBR or VBR-encoded. 
\item It is generic enough to cover a range of extant and new video quality
models. 
\item It explicitly takes into account the constraints of bounded client
buffer and finite horizon. 
\end{itemize}
The solution is generic and flexible enough to cover both video-on-demand
and live streaming scenarios. Our future work includes building an
end-to-end system to deliver video with consistent quality for large-scale
HAS deployments.

{\bibliographystyle{plain}
\bibliography{leeoz,ctech2012}
 } 
\end{document}